\documentclass[aps, reprint, amsfonts,amsmath,amssymb,eufrak]{revtex4-2}

\usepackage{epsfig}
\usepackage{color}
\usepackage[dvipsnames]{xcolor}
\usepackage{bm}
\usepackage{hyperref}
\hypersetup{
    bookmarksnumbered=true, 
    unicode=false, 
    pdfstartview={FitH}, 
    pdftitle={}, 
    pdfauthor={}, 
    pdfsubject={}, 
    pdfcreator={}, 
    pdfproducer={}, 
    pdfkeywords={}, 
    pdfnewwindow=true, 
    colorlinks=true, 
    linkcolor=NavyBlue, 
    citecolor=NavyBlue, 
    filecolor=NavyBlue, 
    urlcolor=NavyBlue 
}

\usepackage[lmargin=.7in,rmargin=.7in,tmargin=.7in,bmargin=1in]{geometry}

\usepackage{braket}

\usepackage{graphicx}
\usepackage{amsthm}

\newcommand{\eq}[1]{\begin{equation} #1 \end{equation}}
\newcommand{\eqa}[2]{\begin{equation} #1 \label{#2} \end{equation}}

\newcommand{\bcases}[1]{\begin{cases} #1 \end{cases}}

\newcommand{\figin}[4]
{\begin{figure}[tb]
\centering
\includegraphics[width= #1]{#2.pdf}
\caption{#3}
\label{f:#4}
\end{figure}}

\newcommand{\todayd}{\the\year/\the\month/\the\day}

\newcommand{\lb}{\label}
\newcommand{\nt}{\notag}
\newcommand{\Tr}{\mathrm{Tr}}

\newcommand{\bel}{\begin{easylist}}
\newcommand{\eel}{\end{easylist}}

\newcommand{\be}[1]{\begin{enumerate} #1 \end{enumerate}}

\newcommand{\eref}[1]{Eq.~\eqref{#1}}

\newcommand{\fref}[1]{Fig.~\ref{f:#1}}

\def\({\left(}
\def\){\right)}
\def\[{\left[}
\def\]{\right]}

\newcommand{\abs}[1]{\left|#1\right|}

\newcommand{\sumtwo}[2]%
{\mathop{\sum_{#1}}_{#2}}
\newcommand{\sumthree}[3]%
{\mathop{\mathop{\sum_{#1}}_{#2}}_{#3}}
\newcommand{\sumfour}[4]%
{\mathop{\mathop{\mathop{\sum_{#1}}_{#2}}_{#3}}_{#4}} 
\newcommand{\prodtwo}[2]%
{\mathop{\prod_{#1}}_{#2}}
\newcommand{\mintwo}[2]%
{\mathop{\min_{#1}}_{#2}}
\newcommand{\maxtwo}[2]%
{\mathop{\max_{#1}}_{#2}}
\newcommand{\maxthree}[3]%
{\mathop{\mathop{\max_{#1}}_{#2}}_{#3}}
\newcommand{\limtwo}[2]%
{\mathop{\lim_{#1}}_{#2}}
\newcommand{\suptwo}[2]%
{\mathop{\sup_{#1}}_{#2}}
\newcommand{\supthree}[3]%
{\mathop{\mathop{\sup_{#1}}_{#2}}_{#3}}
\newcommand{\supfour}[4]%
{\mathop{\mathop{\mathop{\sup_{#1}}_{#2}}_{#3}}_{#4}} 
\newcommand{\inftwo}[2]%
{\mathop{\inf_{#1}}_{#2}}
\newcommand{\infthree}[3]%
{\mathop{\mathop{\inf_{#1}}_{#2}}_{#3}}
\newcommand{\inffour}[4]%
{\mathop{\mathop{\mathop{\inf_{#1}}_{#2}}_{#3}}_{#4}} 

\newcommand\calC{{\cal C}}
\newcommand\calD{{\cal D}}
\newcommand\calE{{\cal E}}

\newcommand\calI{{\cal I}}

\newcommand\calM{{\cal M}}

\newcommand\calP{{\cal P}}
\newcommand\calQ{{\cal Q}}

\newcommand\calS{{\cal S}}
\newcommand\calT{{\cal T}}
\newcommand\calU{{\cal U}}

\newcommand{\bsc}{\boldsymbol{c}}

\newcommand{\bbE}{\mathbb{E}}

\newcommand{\bbN}{\mathbb{N}}
\newcommand{\bbO}{\mathbb{O}}
\newcommand{\bbP}{\mathbb{P}}

\newcommand{\bbZ}{\mathbb{Z}}
\newcommand{\ep}{\varepsilon}

\newcommand{\Di}{\mathit{\Delta}}

\newtheorem{thm}{Theorem}
\newtheorem*{thm*}{Theorem}
\newtheorem{lm}[thm]{Lemma}
\newtheorem{pro}[thm]{Proposition}

\newcommand{\bthm}[1]{\begin{thm} #1 \end{thm}}
\newcommand{\blm}[1]{\begin{lm} #1 \end{lm}}

\theoremstyle{definition}
\newtheorem{dfn}[thm]{Definition}

\usepackage{mathtools}

\newcommand{\bal}{\begin{equation}\begin{aligned}}
\newcommand{\eal}{\end{aligned}\end{equation}}
\newcommand{\ketbra}[2]{|{#1}\rangle\!\langle{#2}|}
\newcommand{\dm}[1]{\ketbra{#1}{#1}}
\DeclareMathOperator{\cl}{cl}

\DeclareMathOperator{\Gibbs}{Gibbs}
\newcommand{\sbar}{\;\rule{0pt}{9.5pt}\right|\;}
\newcommand{\lset}{\left\{\left.}
\newcommand{\rset}{\right\}}

\usepackage{cleveref} 
\crefname{dfn}{definition}{definitions}
\Crefname{dfn}{Definition}{Definitions}
\crefname{thm}{theorem}{theorems}
\Crefname{thm}{Theorem}{Theorems}
\crefname{lm}{lemma}{lemmas}
\Crefname{lm}{Lemma}{Lemmas}
\crefname{pro}{proposition}{propositions}
\Crefname{pro}{Proposition}{Propositions}

\makeatletter
\renewcommand{\@cite}[1]{\textsuperscript{#1)}}
\makeatother

\begin{document}

\preprint{APS/123-QED}



\title{Recovery of the second law in fully quantum thermodynamics}

\author{
{Naoto Shiraishi}
}
\email{shiraishi@phys.c.u-tokyo.ac.jp}
\affiliation{
{Department of Basic Science, The University of Tokyo, 3-8-1 Komaba, Meguro-ku, Tokyo, 153-0041, Japan} 
}%


\author{
{Ryuji Takagi}
}
\email{ryujitakagi.pat@gmail.com}
\affiliation{
{Department of Basic Science, The University of Tokyo, 3-8-1 Komaba, Meguro-ku, Tokyo, 153-0041, Japan} 
}%


\begin{abstract}
Quantum thermodynamics investigates how robust the second law of thermodynamics serves as the unique fundamental law in the small quantum world.
To tackle this problem, the quantum coherence constitutes a major difficulty of investigations, which provides severe constraints hindering the recovery of a single thermodynamic potential.
Here we solve this long-standing problem of quantum information theory by revealing that the state convertibility under thermal operations is fully characterized by the second law of thermodynamics.
Specifically, we prove that whether a quantum state with quantum coherence is convertible to another by a thermal operation with a correlated catalyst is completely determined by the free energy ordering.
Unlike previous attempts, our setting does not resort to any additional external coherent assist, providing a faithful operational characterization of thermodynamic state transformation.
\end{abstract}

\maketitle

The second law of thermodynamics is one of the most fundamental and universal principles in our physical world.
The second law manifests the entropy (adiabatic case) or the free energy (isothermal case) as the universal unique criterion of state convertibility of macroscopic systems and excludes the presence of other fundamental restrictions.
An axiomatic thermodynamics~\cite{LY99} derives the presence of the aforementioned second law in the macroscopic world, where the importance of the uniqueness of the thermodynamic potential is argued in detail.

Quantum thermodynamics investigates the validity of the second law in the small quantum world~\cite{Gour-review, Los-review, Sag-review}.
In this field, a thermal operation is generally set as a tractable thermodynamic operation, which is implementable by an energy-conserving unitary operation with an auxiliary system in the Gibbs state~\cite{Jan00, HO13, Bra13} (\fref{setup}.(a)).
This bottom-up definition provides a clear operational foundation of this operation class, which is a reason why the thermal operation is preferred in quantum thermodynamics.

The thermal operation is, however, known to have a limited conversion power, where many constraints other than the second law naturally arise and the second law is buried under them~\cite{HO13, Gou18}.
To enhance convertibility and recover the second law, the addition of an external auxiliary system called a catalyst, which does not change its own state but helps the conversion of the state in the system~\cite{Bra15,Mul18,Dat22, BWN23, Shi25}, is widely used as a standard approach.
In particular, we focus on a correlated catalyst, where the system and the catalyst can have a negligibly small correlation in the final state (\fref{setup}.(b)). 
Notably, if we restrict possible states to energy-diagonal states, the second law is indeed recovered in this setting~\cite{Mul18}.

In contrast to the above success, the situation becomes much more difficult for general quantum states.
A major obstacle is arisen from the presence of quantum coherence among energy eigenstates, which cannot be created from scratch by any energy-conserving operation including thermal operations~\cite{Mar-thesis, FOR15, TT25}.
This barrier again leads to emergence of various constraints other than the second law~\cite{NG15, Cwi15, Los15}.
The difficulty of handling quantum coherence is highlighted by the coherence no-broadcasting theorem, establishing that an incoherent initial state is never transformed into a coherent state by any energy-conserving operation even with the help of a correlated catalyst~\cite{LM19, MS19}.
On the basis of these observations, most of existing studies on quantum thermodynamics with thermal operations reluctantly resort to ad hoc external coherent assists supplying some coherence from outside~\cite{Bra13, SSP14, Fai19,Gou22}.
Some other existing studies give up treating thermal operations and instead employ a larger class of operations known as Gibbs-preserving operations~\cite{SS21}, though a recent study suggests that this relaxation may incur unfeasible thermodynamic resources to implement~\cite{TT25}.
Overall, previous studies indicate that the recovery of the second law with thermal operations in the fully quantum setting is prohibitively difficult.

Here, in contrast to the previous pessimistic view, we solve this long-standing problem by establishing the second law of thermodynamics as a single inequality in terms of the free energy.
This provides the necessary and sufficient condition of state conversion by a thermal operation with a correlated catalyst, which puts positive answers to the conjectures~\cite{LM19, GKS24}.
We specifically prove that a coherent quantum state $\rho$ is convertible to $\rho'$ by a thermal operation with a correlated catalyst if and only if $\rho$ has a larger free energy than $\rho'$, which declares quantum thermodynamics with thermal operations as a reversible theory.
Unlike previous literature, we do not adopt any external coherence assist that supplies small but finite amount of coherence to the system without restoration.
Instead, we extract the required coherence only from the given state $\rho$ and the catalyst, not relying on the external coherence supply and thus making the setting operationally relevant and faithful. 
As a result, the second law is fully recovered for all coherent states.

\figin{8.7cm}{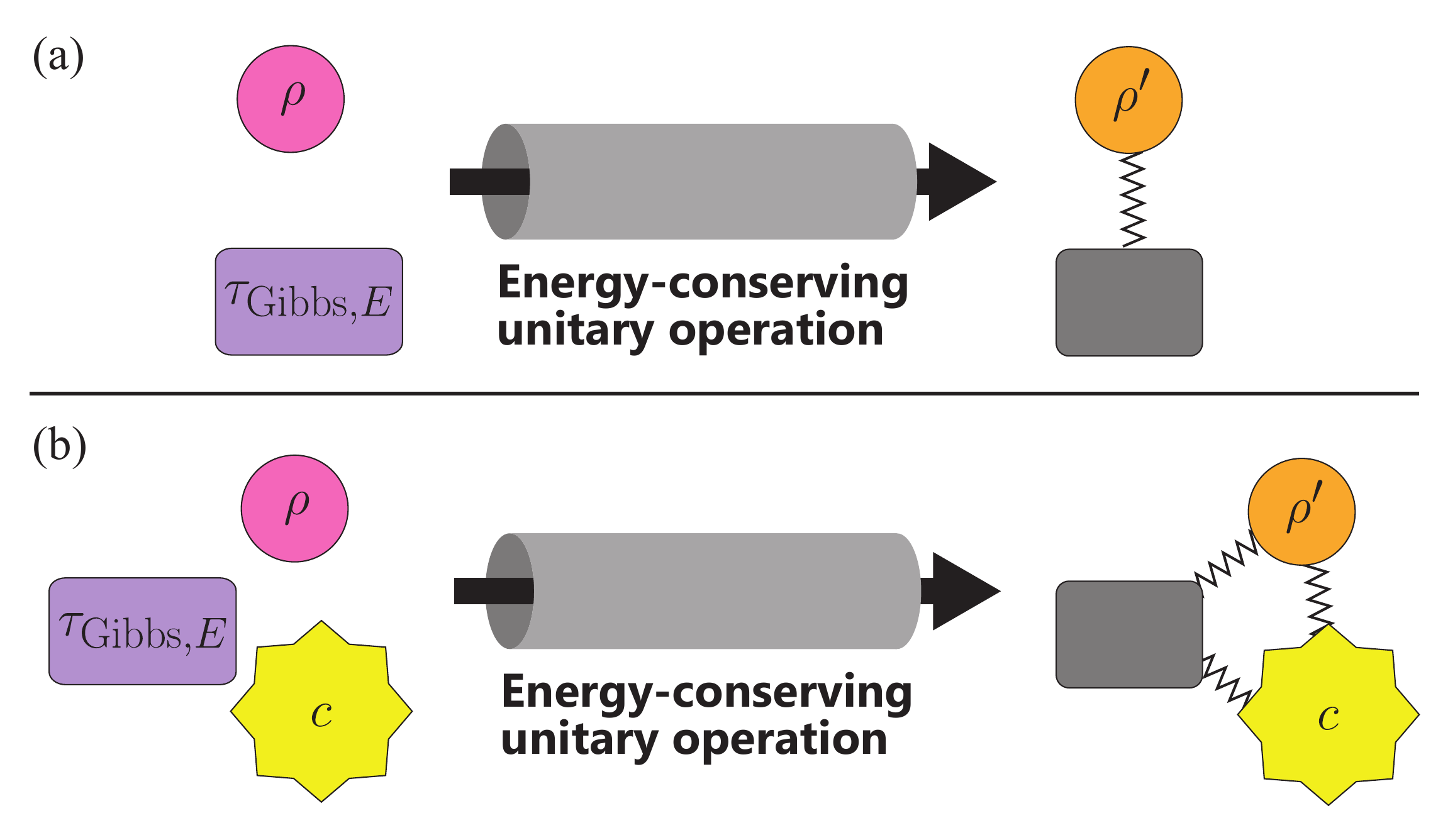}{
{\bf Thermal operation with and without a correlated catalyst.a-b.}
{\bf (a)} 
In a thermal operation, we employ an auxiliary systems; a thermal environment in its Gibbs state $\tau_{\Gibbs , E}$.
We apply an energy-conserving unitary operation on them.
The final reduced state of the system is $\rho'$.
{\bf (b)} 
In a thermal operation with a correlated-catalyst, we employ a further external auxiliary system; a catalyst in state $c$, which is arbitrary but should return to its original state.
We apply an energy-conserving unitary operation on the system and two external systems.
The final reduced state of the system is $\rho'$ and that of the catalyst returns to $c$.
}{setup}

\bigskip

{\bf Thermal operation with a catalyst} --- In quantum thermodynamics, we study feasible state transformations with the available thermodynamic operations. 
It has turned out that the resource-theoretic approach~\cite{Chitambar2019quantum} is useful to this end, as it allows us to formulate the operational setting in a rigorous manner and offers a powerful technique developed in quantum information theory.

To start our investigation, we first need to formalize the available thermodynamic operations.
Here, the standard and operationally natural choice for the set of available operations is {\it thermal operations}, which models the class of operations implemented freely in an isothermal condition, as tractable thermodynamic processes in the small quantum regime.
Here, a thermal operation $\Lambda$ is constructed by an energy-conserving unitary $U$ on system $S$ and external system $E$ whose initial state is in the Gibbs state as $\Lambda(\rho)=\Tr_E[U (\rho\otimes \tau_{\Gibbs,E}) U^\dagger]$ (\fref{setup}.(a))~\cite{HO13,Bra13,Sagawa2021asymptotic-reversibility,Fai19}.
The thermal operation is defined in a bottom-up fashion, combining a freely preparable state (a Gibbs state) and a thermodynamically implementable operation (a unitary operation satisfying the law of energy conservation), which is well supported from the operational viewpoint.
This is why the thermal operation is regarded as the standard choice of a set of tractable thermodynamic operations.

The thermal operation itself does not have much state convertibility, which motivates us to enhance the convertibility by introducing some external aids with keeping thermodynamic picture.
We here employ one of the standard framework in resource theories, a {\it catalytic framework}, where we employ an auxiliary system called {\it catalyst} $C$ which does not change its own state while it helps state conversions in the system (\fref{setup}).
Although the presence of catalyst may initially appear to play almost no role since the state of the catalyst does not change at all, in fact the catalyst leads to a significant change in the convertibility.
The idea of the catalyst has already been seen in the axiomatic characterization of thermodynamics~\cite{LY99}, where we examine state convertibility without leaving any change in the rest of the universe.
In other words, we can borrow arbitrary states from somewhere in the universe as long as we return the state exactly the same as the borrowed one.
In this article, we employ a {\it correlated catalyst}~\cite{Dat22, BWN23, Shi25, Mul18, SS21, GKS24} whose reduced state at the final stage returns to the original state exactly but may have a small correlation between the system and the catalyst.
In this setting, by hiding the main system, an observer in the rest of the universe cannot find any change in the catalyst, which ensures the catalytic nature of this setting.

In summary, we examine whether conversion $\rho\to \rho'$ can be implemented by energy-conserving unitary $U$ on $S\otimes E\otimes C$ as $\pi=U(\rho\otimes \tau_{\Gibbs,E}\otimes c)U^\dagger$ satisfying $\Tr_{EC}[\pi]=\rho'$ and $\Tr_{SE}[\pi]=c$, where $\tau_{\Gibbs,E}$ refers to the thermal Gibbs state of $E$.
In this implementation, an external system in the Gibbs state $\tau_{\Gibbs,E}$ can be discarded, while an external system not in the Gibbs state should return to its original state, which reflects the fact that a Gibbs state is the only state we can freely prepare.
If we allow a {\it vanishing error} in the final state of the system, the condition $\Tr_{EC}[\pi]=\rho'$ is replaced by $\|\Tr_{EC}[\pi]-\rho'\|_1<\ep$ for any given $\ep>0$, which is also stated as a conversion {\it with arbitrary accuracy}.
Note that we do not allow any small error in the catalyst, since a catalyst with vanishing error yields unphysical consequences that any state transformation is possible~\cite{Bra15}.

\bigskip
{\bf Difficulty of quantum thermodynamics} --- 
We here clarify what has already been known and what has been left unknown in existing literature, which elucidates why quantum thermodynamics is difficult.

Many efforts have been invested to see whether the second law---formalized by the single inequality on the Helmholtz free energy---can be recovered under the thermal operations.
We have now gained a good understanding if we restrict possible states to \emph{quasi-classical states}, which are the states whose density matrices have a block-diagonal structure in the energy eigenbasis.  
In particular, the second law as a single inequality with respect to the Helmholtz free energy is recovered in the asymptotic~\cite{Bra13} and the correlated-catalytic~\cite{Mul18}  framework, for the case where both the initial state and the target state are quasi-classical.

Unlike the aforementioned success achieved under restricted states, the characterization of general quantum states remains far from a comprehensive understanding.
This challenge stems from the presence of \emph{quantum coherence}, i.e., nonzero off-diagonal terms representing the quantum superposition among energy eigenstates.
This is a fundamental problem that cannot be avoided in quantum thermodynamics, because quantum superposition represents the unique quantum feature that does not appear in the classical setting. 
The technical difficulty in dealing with the fully quantum setting lies in the fact that thermal operations cannot create coherence from scratch, and thus if the target state contains nonzero coherence, it appears inevitable to get coherence supply from somewhere else. 
This obstacle led Refs.~\cite{Bra13,Gou22,Fai19,Sagawa2021asymptotic-reversibility} to relax the setting of asymptotic transformation to the one where external supply of coherence over a sublinear range of energy levels is allowed, in order to recover the second law in this restricted setting. 
This external coherence, however, is an additional resource supplied from the outside and thus these results do not directly resolve the issue of the second law in the fully operational sense.

The difficulty stemming from quantum coherence motivates another line of attempts involving catalysts.
Indeed, Ref.~\cite{Mul18}, which showed the second law for quasi-classical states, first conjectured that the second law could be extended to a fully quantum setting with the help of correlated catalysts. 
However, the possibility of direct extension to the fully quantum setting was defeated by the coherence no-broadcasting theorem~\cite{LM19,MS19}, and reflecting this finding the above conjecture was withdrawn and instead the same author conjectured the recovery of the second law for Gibbs-preserving operations~\cite{LM19}.
Here, the set of Gibbs-preserving operations is strictly larger than the set of thermal operations, where thermal operations are contained as a proper subset.
This modified problem was settled by Ref.~\cite{SS21}, which showed that the second law is recovered under the Gibbs-preserving operations with correlated catalysts.
However, given that Gibbs-preserving operations can create coherence from scratch~\cite{FOR15} and come with a significant coherence cost for its implementation~\cite{TT25}, this result does not directly address the core problem either.

Overall, the question of whether one can characterize the general state convertibility under thermal operations (potentially in the setting of asymptotic or catalytic transformation) by the single inequality of free energy without the external coherence supply has been left wide open.

\bigskip
{\bf The second law of thermodynamics as a reversible theory} --- Before stating our main result, we shall introduce some concepts which we will use to characterize whether an initial state has nonzero coherence.
Let us define the set of \emph{coherent modes} $\calD(\rho)$ and \emph{resonant coherent modes} $\calC(\rho)$ defined by
$\calD(\rho):=\lset \Di_{ij} \sbar \Di_{ij}=E_i-E_j,\ (i,j)\in \calM(\rho)\rset$ and $\calC(\rho):=\lset x  \sbar x=\sum_{\Di\in \calD(\rho)} a_\Di \Di,\ a_\Di\in \bbZ\rset$,
where $\calM(\rho)$ is the set of integer pairs $(i,j)$ such that there exists a pair $\ket{E_i}$, $\ket{E_j}$ of energy eigenstates with energies $E_i$ and $E_j$ satisfying $\braket{E_i|\rho|E_j}\neq 0$, and $\bbZ$ denotes the set of integers. 
Namely, $\calD(\rho)$ is the energy difference (mode) on which $\rho$ has nonzero coherence, and $\calC(\rho)$ is the set of modes that can be obtained as an integer-linear combination of the coherent modes. 
We employ this slightly complicated concept because it is known that we can merge coherent modes (e.g., we can convert a state with coherent modes 1 and $\sqrt{2}$ to a state with coherent mode $1+\sqrt{2}$) with the help of a correlated catalyst~\cite{ST23}.
We stress that these sets are defined regardless of the \emph{amount} of coherence $\rho$ has:
What only matters is whether $\rho$ has nonzero coherence on those modes or not. 

Our result shows that, as long as all coherent modes of the target state $\rho'$ is also coherent in the initial state $\rho$ in terms of the resonant coherent modes, the feasible transformation is completely characterized by the free energy. (See Method for a proof sketch and Supplemental Material for its detailed proof.)

\begin{thm*}
For arbitrary two states $\rho$ and $\rho'$, suppose that $\rho$ has resonant coherent modes of $\rho'$ with integer coefficients, i.e., $\calC(\rho')\subseteq\calC(\rho)$.
Then, we can convert $\rho$ to $\rho'$ by a thermal operation with a correlated catalyst with arbitrary accuracy if and only if the free energy ordering 
\eq{
F(\rho)\geq F(\rho')
}
is satisfied, where $F(\rho)=\Tr(\rho H) - S(\rho)/\beta$ is the nonequilibrium free energy with Hamiltonian $H$ and inverse temperature $\beta$.
In addition, the correlation between the system and the catalyst can be set arbitrarily small.
\end{thm*}

Our result shows that, given the coherence condition $\calC(\rho')\subseteq\calC(\rho)$, the state transformation is fully characterized by the single inequality on Helmholz free energy, recovering the second law in the fully quantum scenario.
This theorem provides an answer to an outstanding problem in quantum information theory and quantum thermodynamics, asking what is the necessary and sufficient condition of state conversions by thermal operations and how the second law of thermodynamics is recovered.
In particular, we prove the conjecture raised in Ref.~\cite{Mul18} for states with finite coherence and the conjecture raised in Ref.~\cite{KGS23} with a slight modification on the coherence condition.

We emphasize that, unlike previous attempts, our recovery of the second law of thermodynamics with thermal operations is fulfilled without any external coherence assist, and all systems are treated in a closed fashion.
Contrasting our result with that in Ref.~\cite{SS21} leads to an important consequence that thermal operations and Gibbs-preserving operations have essentially the same state convertibility, which means the collapse of these two resource theories in the correlated-catalytic framework.
This result is striking in view of the fact that thermal operations and Gibbs-preserving operations have significantly different state convertibility without catalyst, whose gap cannot be closed by consuming finite amount of work~\cite{TT25}.

\bigskip
{\bf Restriction from quantum coherence} --- 
We stress that our coherence condition $\calC(\rho')\subset \calC(\rho)$ is extremely mild:
If the initial state $\rho$ has even a tiny amount of coherence on some modes, then the target state can have any amount of coherence on this mode. 
Consequently, even a small perturbation to the initial state completely removes the restriction on the coherent modes. 
To put this in another way, the state such that $\calC(\rho)$ contains all modes dominates the state space with measure 1, and thus for any $\rho'$ randomly sampled $\rho$ almost surely satisfies the coherence condition.

Recall that we must have conditions on coherence because of the coherence no-broadcasting theorem~\cite{LM19, MS19}, which prohibits a conversion from an incoherent state to a coherent state regardless of the free energy ordering. 
Thus, our condition is essentially the best we could hope for.
It also shows an interesting instability in the coherence no-broadcasting theorem that we cannot obtain any coherent state from the incoherent state even with the help of a correlated catalyst, but the restriction is discontinuously switched to the free energy ordering once the initial state contains nonzero coherence.

\bigskip
{\bf Reversible theory} --- 
Our main theorem establishes a {\it reversible theory} in quantum thermodynamics.
Here, a theory is reversible if we can prepare a resource storage system $A$ with the following property: For any $\rho$ and $\rho'$, there exists a state $\alpha$ and $\alpha'$ on $A$ such that $(\rho, \alpha)$ can be transformed to $(\rho', \alpha')$ and vice versa.
In our case, any (sufficiently large) system serves as a resource storage system by choosing states $\alpha$ and $\alpha'$ as satisfying $F(\rho)-F(\rho')=F(\alpha')-F(\alpha)$.

The reversibility is a desirable property since it declares that all resources can be completely exhausted without any loss by storing residual resources in the resource storage system.
We note that many resource theories including the entanglement theory~\cite{Lami2023no} lack reversibility and thus generally have inevitable losses in a composition of operations.
In addition to this point, a reversible theory leads to a complete ordering by a single monotonic resource measure, which is another advantage of reversibility from the theoretical aspect.
In light of these facts, the reversibility in quantum thermodynamics is highly nontrivial and desirable.

\figin{8.7cm}{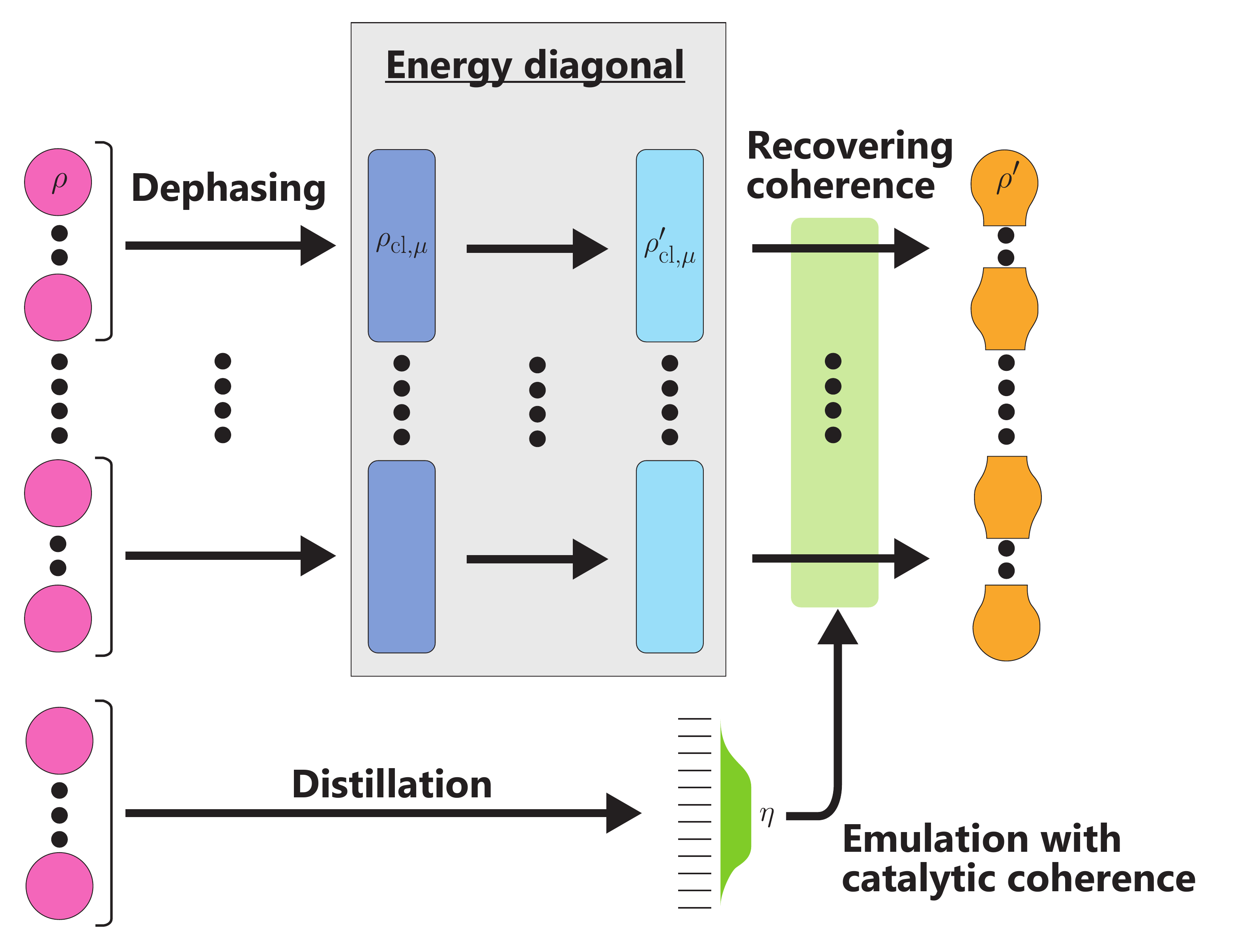}{
{\bf Overview of the protocol.}
We first divide $n$ copies of the initial state $\rho$ into $\mu\nu$ copies and $o(n)$ copies.
Using the dephasing channel, we remove off-diagonal elements of each $\mu$ copies of $\rho$ and obtain energy-diagonal state $\rho_{\cl,\mu}$ on $S^{\otimes \mu}$.
We then convert $\nu$ copies of $\rho_{\cl,\mu}$ into another energy-diagonal state $\rho_{\cl,\mu}'$ with the same number of copies.
The state $\rho_{\cl,\mu}'$ has the same spectrum as $\rho'^{\otimes \mu}$.
Separately from this process, we distill state $\eta$ with broad coherence from $o(n)$ copies of $\rho$.
We finally recover $\rho'^{\otimes \mu}$ from $\rho_{\cl,\mu}'$ by an energy-conserving unitary operation by using $\eta$ as catalytic coherence repeatedly.
}{schematic}

\bigskip
{\bf Conversion protocol} --- 
We here explain several key ideas to construct a correlated-catalytic protocol converting $\rho$ to $\rho'$ if the free energy ordering is satisfied.
We first observe the technique for constructing a correlated-catalytic transformation from an asymptotic conversion protocol~\cite{Duan2005multiple-copy,SS21,TS22,ST23,GKS24}.
However, we emphasize that the conventional asymptotic conversion is precluded in our setting, since the coherence no-distillation theorem~\cite{marvian2020coherence} prohibits the asymptotic transformation from a full rank state $\rho$ to a pure coherent state $\rho'$ regardless of their free energy ordering. 
To avoid this difficulty, we need to employ a modified version of an asymptotic conversion, known as {\it asymptotic marginal conversion}~\cite{Ferrari2023asymptotic,GKS24,ST23}, which aims to obtain a state whose local state is close to the target state. 
We then employ the construction of the desired correlated-catalytic transformation from an asymptotic marginal transformation, which is recently introduced in Refs.~\cite{ST23, GKS24}.

This allows us to focus on constructing the asymptotic marginal transformation with the unit rate for the initial and target states $\rho$ and $\rho'$ with $F(\rho)\geq F(\rho')$.
The core idea underlying our construction is to first obtain the classical state that has approximately the same eigenvalues and free energy of $\rho'$, and then rotate its eigenbasis to the coherent eigenbasis of $\rho'$ (\fref{schematic}).
Recalling that the second law is recovered for classical states~\cite{Bra13,HO13,Shiraishi2020twoconstructive}, we conclude that such a classical state can be obtained from $\rho$ under the free energy ordering $F(\rho)\geq F(\rho')$. 
Therefore, our problem is reduced to the feasibility of the coherence-generating unitary that rotates the classical basis to the coherent eigenbasis of the target state.

\figin{6cm}{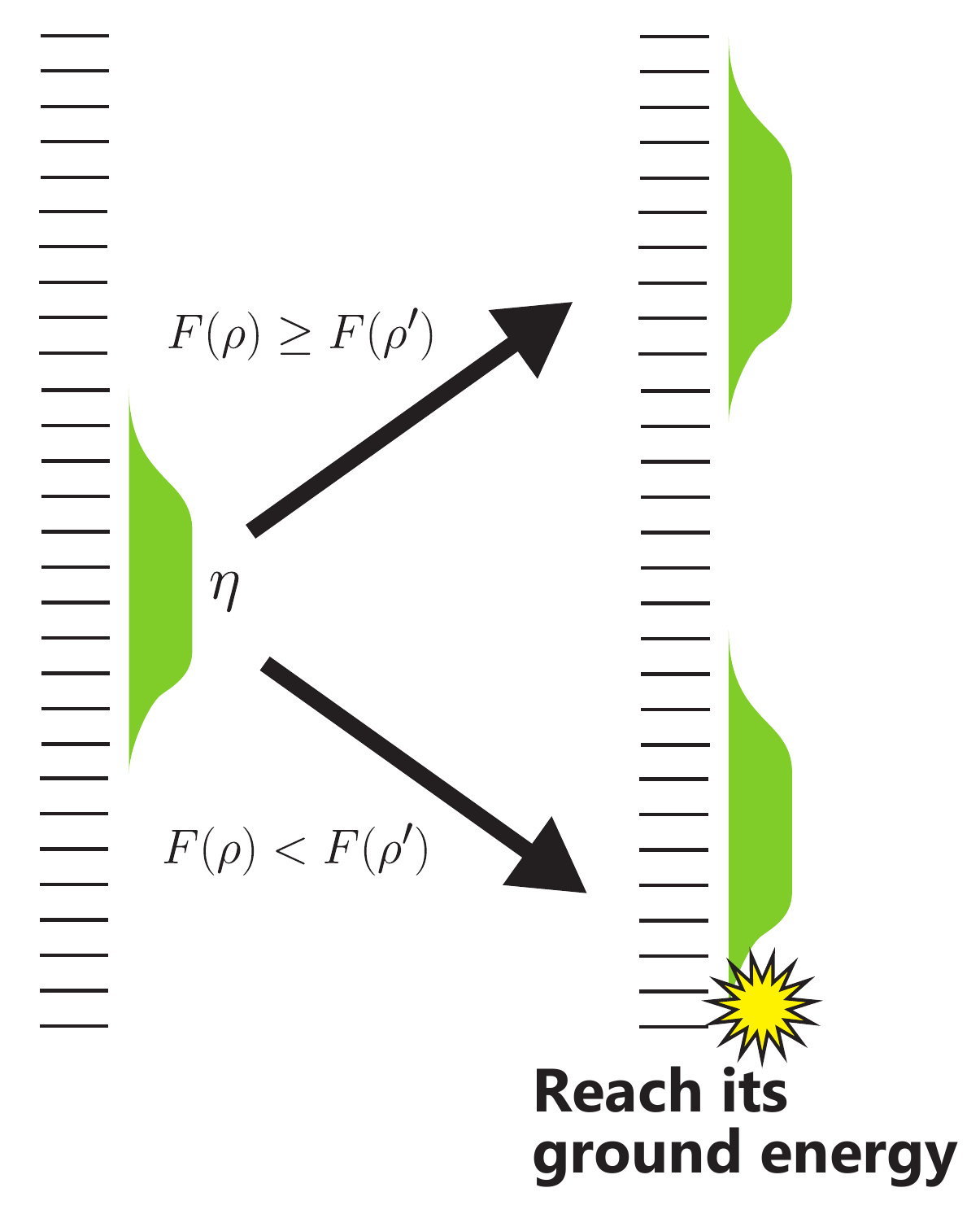}{
{\bf Random walks modeling the energy distribution of the coherent resource state.}
If the free energy ordering $F(\rho)\geq F(\rho')$ is satisfied, the expected energy of the coherent resource state increases, corresponding to the random walk with moving upward on average. In this case, the probability of hitting the ground energy can be arbitrarily suppressed. On the other hand, if the free energy ordering is violated, i.e., $F(\rho)<F(\rho')$, then the expected energy of the coherent resource state decreases by this process, resulting in hitting the ground energy with high probability and causing undesirable error.
}{ladder}

We show that the coherence-generating unitary is implementable by using only an additional \emph{constant} number of copies of the initial state $\rho$ with respect to the number of copies of $\rho$ in the initial state of the asymptotic marginal transformation.
Notice that extra $O(1)$ (or more generally, $o(n)$) copies of the initial state $\rho$ does not affect the asymptotic transformation rate.
We stress that this does \emph{not} mean that we use an external coherence resource, and this extra use of the initial state can be properly taken into account in constructing the catalyst. 
At the first sight, it might look implausible that only the constant number of the initial state $\rho$ allows us to apply the desired coherent unitary, \emph{which needs to be applied to a growing number of subsystems}, as our goal is to obtain $n$ copies of the state $\rho'$ in the local subsystem. 
We accomplish this by \emph{reusing} the coherence extracted from the initial state $\rho$ over and over to apply the coherent unitary repeatedly. 
The main technique we employ is \emph{catalytic coherence} introduced in Ref.~\cite{Abe14}, which provides a way of reusing the coherent resource state for the unitary implementation.

However, we cannot directly apply the catalytic coherence protocol in our setting, since our coherent resource state sits in a system with a ground energy.
If the state might eventually hit the ground state, we can no longer perform the unitary implementation (\fref{ladder}).  
Fortunately, if the free energy ordering $F(\rho)\geq F(\rho')$ holds, the probability of the hitting event to the ground state can be arbitrarily suppressed by initially pumping \emph{constant} amount of energy, which implies an accurate implementation. 
We show this by mapping the behavior of the resource state into classical \emph{random walks} and estimate the probability of hitting the ground state using the martingale theory. 
Here, the free energy ordering $F(\rho)\geq F(\rho')$ ensures that the random walk moves upward on average, which allows us to suppress the hitting probability during an \emph{infinite} number of applications of coherent unitary.
In contrast, without the free energy ordering  (i.e., $F(\rho)< F(\rho')$), the dynamics of the coherent resource state are energy-lowering on average, which inevitably results in hitting the ground state.

\bigskip
{\bf Discussion} --- Our result solves a major open problem in quantum information theory, particularly in quantum thermodynamics, by establishing the necessary and sufficient condition of state conversions by thermal operation in the correlated-catalytic framework.
We successfully recover the second law in the small quantum regime, without resorting any external aid on coherence.

Our result claims essentially the same conversion power of thermal operations as that of another important operation class, Gibbs-preserving operations~\cite{Gour-review, Los-review}, which is a strictly wider class than the thermal operation.
This result shows clear contrast to a recent study implying a big discrepancy between thermal operations and Gibbs-preserving operations without catalyst~\cite{TT25}.
We also establish the collapse of thermal operations and {\it enhanced thermal operations}~\cite{Shi25-2}, which is the intersection of Gibbs-preserving operations and energy-preserving operations, in the correlated catalytic framework.
Although the enhanced thermal operation has been proven to be a strictly larger class than the thermal operation without catalyst~\cite{Cwi15, DDH21}, these two provide the same theory with a correlated catalyst.

We finally comment on how complex our operation is.
In principle, a thermal operation may require a complicated energy-conserving unitary operations with an elaborated auxiliary system in the Gibbs state.
However, it has been discovered that in the correlated catalytic framework several more restricted and tractable operational classes, two-body interactions and Markovian dynamics, have the same state conversion power as thermal operations~\cite{SN23}.
Hence, it suffices to prepare weaker classes of operations, such as two-body interactions and Markovian dynamics, to implement the desired thermal operation.

\bigskip

{\bf Method}

{\footnotesize

\medskip

{\bf Thermal operation}

In quantum thermodynamics, a set of free operations which we can implement freely is thermal operations.
A channel $\Lambda$ from $S$ to $S'$ is called a thermal operation if there is an external system $E$ with Hamiltonian $H_E$, the Gibbs state $\tau_{\Gibbs,E} = e^{-\beta H_E}/\Tr(e^{-\beta H_E})$, and energy-conserving partial isometry $V_{SE\to S'E}$ with $V_{SE\to S'E}(H_S+H_E)=(H_{S'}+H_E)V_{SE\to S'E}$ satisfying~\cite{Sagawa2021asymptotic-reversibility} 
\bal
 \Lambda(\rho) = \Tr_E\left[V_{SE\to S'E}(\rho\otimes \tau_{\Gibbs,E}) V^\dagger_{SE\to S'E}\right].
\eal

The thermal operation belongs to two important classes of operations.
One is Gibbs-preserving operations.
A channel $\Lambda$ on $S$ is called a Gibbs-preserving operation if the channel satisfies
\eq{
\Lambda(\tau_{\Gibbs,S})=\tau_{\Gibbs,S},
}
where $\tau_{\Gibbs,S} = e^{-\beta H_S}/\Tr(e^{-\beta H_S})$ is the Gibbs states. 
It is known that if states are block-diagonal in the energy basis (i.e., quasi-classical), then thermal operations and Gibbs-preserving operations have the same power of state transformations~\cite{HO13, Shiraishi2020twoconstructive}, i.e., for two quasi-classical states $\rho$ and $\rho'$ any $\ep>0$ there is a thermal operation $\Lambda_{\rm TO}$ satisfying $\|\Lambda_{\rm TO}(\rho)-\rho'\|_1<\ep$ if and only if there is a Gibbs-preserving operation $\Lambda_{\rm GPO}$ satisfying $\|\Lambda_{\rm GPO}(\rho)-\rho'\|_1<\ep$.
On the other hand, for general quantum states with quantum coherence, thermal operations is a strict subset of Gibbs-preserving operations due to the covariant property explained below~\cite{FOR15}.
Moreover, the coherence cost compensating this gap is shown to diverge in general~\cite{TT25}, which implies a marked discrepancy between these two classes in general quantum settings.

The other important class which thermal operations belong to is covariant operations.
A channel $\Lambda$ from $S$ to $S'$ is called a covariant operation if 
\bal
 \Lambda(e^{-iH_S t}\rho e^{iH_S t}) = e^{-iH_{S'}t}\Lambda(\rho) e^{iH_{S'} t}.
\eal
is satisfied for any $t$.
Note that we have an equivalent operational definition of covariant operations~\cite{marvian2020coherence}.
A channel $\Lambda$ from $S$ to $S'$ is a covariant operation if there is an external system $E$ with Hamiltonian $H_E$ and its block-diagonal state in the energy basis $\kappa_{\rm diag}$, and energy-conserving isometry $V_{SE\to S'E}$ with $V_{SE\to S'E}(H_S+H_E)=(H_{S'}+H_E)V_{SE\to S'E}$ satisfying 
\bal
 \Lambda(\rho) = \Tr_E\left[V_{SE\to S'E}(\rho\otimes \kappa_{\rm diag}) V^\dagger_{SE\to S'E}\right].
\eal
Thus, a covariant operation can be obtained from the definition of a thermal operation by replacing the Gibbs state $\tau_{\Gibbs,E}$ by a general classical state $\kappa_{\rm diag}$.

These two properties, the Gibbs-preserving property and the covariant property, are considered to be two key characterizations of thermal operations~\cite{Los-review}.
However, it is also revealed that the set of thermal operations is still a strict subset of the intersection of Gibbs-preserving operations and covariant operations, which is also called {\it enhanced thermal operations}, even in the approximate sense~\cite{DDH21}.
Namely, we have a thermal operation $\Lambda_{\rm TO}$ transforming $\rho$ to $\rho'$ such that there is no sequence of covariant Gibbs-preserving operations $\{ \Lambda_{{\rm covGPO},m}\}_m$ satisfying $\lim_{m\to \infty}\Lambda_{{\rm covGPO},m}(\rho)=\rho'$.
This result implies a (maybe subtle but) missing constraint in thermal operations other than the Gibbs-preserving property and the covariant property.

\bigskip

{\bf Caralytic framework} 

In a catalytic framework, we use an ancillary system which helps the state transformation, while the state of the ancillary system returns to its original state. 
This ancillary system is called {\it catalyst}, and this transformation is called a {\it catalytic transformation}.
We especially focus on the framework with a {\it correlated catalyst}, where the final state may have a correlation between the main system and the catalyst, while the reduced state on the catalyst should be exactly the same state as the initial catalyst state. 
In other words, a state $\rho$ can be transformed to a state $\rho'$ by a thermal operation with a correlated catalyst with a vanishing error if for any $\ep>0$ there exists a state $c$ on a catalyst $C$ and a thermal operation $\Lambda$ such that 
\bal
 \Lambda(\rho\otimes c) = \pi_{SC},\quad \|\Tr_C(\pi_{SC})-\rho'\|_1\leq \epsilon,\quad \Tr_S(\pi_{SC})=c.
\eal

Here, we allow a vanishing error in the system but do not allow any error in the catalyst in order to avoid {\it embezzlement}.
In fact, it is reported that if we allow a vanishing error in the catalyst, then any transformation becomes possible, at least in the quasi-classical regime~\cite{Bra15}, which spoils the second law of thermodynamics.
This is the reason why we conventionally prohibit any error in the catalyst.

\bigskip

{\bf Restriction on coherence manipulation} 

The covariant condition, which respects the preciousness of coherence, accompanies various restrictions on state transformations.
One basic property is the monotonic decrease of coherent modes.
Let $\{ \ket{E_i}\}$ be a set of energy eigenstates with energies $E_i$.
We define the {\it coherent modes} of $\rho$ by a set of differences between two eigenenergies with nonzero off-diagonal elements in terms of the energy eigenbasis:
\eq{
\calD(\rho):=\lset \Di_{ij} \sbar \Di_{ij}=E_i-E_j,\ (i,j)\in \calM(\rho)\rset,
}
where $\calM(\rho)$ is the set of integer pairs $(i,j)$ satisfying $\braket{E_i|\rho|E_j}\neq 0$.
It is shown that if a covariant operation $\Lambda$ transforms $\rho$ to $\rho'$, then $\calD(\rho')\subseteq \calD(\rho)$ holds~\cite{MS14}.
This means that the coherent modes monotonically decay through a covariant operation.

The above monotonicity is shown for single-shot covariant operations, and we may avoid this consequence in some enhanced frameworks including catalytic and asymptotic frameworks.
However, even these enhanced frameworks we still face harsh restrictions on possible conversions.

The first restriction is the {\it coherence no-broadcasting theorem}~\cite{MS19, LM19} on the catalytic framework, claiming that if $\rho$ is an incoherent state on $S$ (i.e., $\calD(\rho)=\{ 0\}$), then for any catalyst $C$ with state $c$ and a covariant operation $\Lambda: S\otimes C\to S'\otimes C$ satisfying
\eq{
\Lambda(\rho\otimes c)=\pi, \hspace{10pt}
\Tr_{S'}[\pi]=c,
}
then we find that $\Tr_C[\pi]$ is also incoherent.
This no-go theorem prohibits a transformation from an incoherent state to a coherent state by a covariant operation even with the help of a correlated catalyst.

Another restriction is the {\it coherence no distillation theorem}~\cite{marvian2020coherence} on the asymptotic framework, claiming that if $\rho$ is a full-rank state (i.e., its matrix representation $\rho_{ij}:=\braket{E_i|\rho|E_j}$ is a full-rank matrix) and $\rho'$ is a pure coherent state, then no covariant operation asymptotically transforms $\rho$ to $\rho'$ with a finite transformation rate.
Namely, for any  $r>0$ there is a sufficiently large $N$ and $\ep>0$ such that for any $n\geq N$ no covariant operation $\Lambda: S^{\otimes n}\to S'^{\otimes \lfloor rn\rfloor}$ satisfies $\| \Lambda(\rho^{\otimes n})-\rho'^{\otimes \lfloor rn\rfloor}\|_1<\ep$.
This no-go theorem essentially prohibits an asymptotic transformation to a pure maximally coherent state, which is usually set as a {\it golden unit} of a resource theory of U(1) asymmetry, with a meaningful transformation rate.

\bigskip

{\bf Asymptotic-catalytic correspondence} 

We first note the deep connection between catalytic conversions and {\it asymptotic conversions}~\cite{Duan2005multiple-copy,SS21,TS22,ST23,GKS24}.
In an asymptotic conversion, we consider conversions of multiple copies of $\rho$ to multiple copies of $\rho'$.
Specifically, we focus on the {\it asymptotic marginal conversion} with conversion rate arbitrarily close to 1 with a vanishing error.
In an asymptotic marginal conversion, we consider states on $S^{\otimes n}$ and convert $\rho^{\otimes n}$ to $\Xi$ on system $S^{\otimes n'}$ with $n'=n-o(n)$ satisfying 
\eq{
\|\rho'-\Tr_{\backslash i}[\Xi]\|_1<\ep
}
for any given $\ep>0$ for $i=1,\dots,n'$.
We emphasize that the definition of the error of the asymptotic marginal conversion is different from that of the conventional asymptotic conversion, which is measured by 
\eq{
\| \rho'^{\otimes n'}-\Xi\|_1<\ep.
}
An asymptotic conversion is always an asymptotic marginal conversion, while the opposite is not always true, since the asymptotic marginal conversion does not care about correlations among copies.
Although this correlation can be set arbitrarily small, this difference leads to substantially different consequences.
In particular, in order to avoid the restriction of no distillation theorem~\cite{marvian2020coherence}, we should employ this tricky setting, the asymptotic marginal conversion, for our proof.

Below we show that the presence of an asymptotic marginal conversion with rate arbitrarily close to 1 implies a correlated-catalytic conversion.
Although a general method for constructing a correlated catalytic transformation from the asymptotic transformation is well known~\cite{Duan2005multiple-copy,SS21,ST23,fang2024surpassingfundamentallimitsdistillation}, we need to slightly modify the protocol in order to apply to our setting, since our asymptotic transformation involves different numbers of copies between the input and output states.
We address this by appropriately placing thermal Gibbs states in our catalyst as discussed in Ref.~\cite{fang2024surpassingfundamentallimitsdistillation}.

To this end, we modify the above thermal operation $\Lambda$ that comes with input and output systems of different sizes to the one with the same input and output systems by attaching the $o(n)$ copies of thermal Gibbs state $\tau_{\Gibbs}^{\otimes o(n) }$ to the output system, which we denote by $\Lambda'$.
Namely, expressing $\Xi\coloneqq \Lambda(\rho^{\otimes n})$, the output from $\Lambda'$ denoted by $\Xi'$ is given by $\Xi'=\Lambda'(\rho^{\otimes n}) = \Xi\otimes \tau_{\Gibbs}^{\otimes o(n)}$.

Now we construct the catalyst of the single-shot correlated-catalytic transformation from $\rho$ to $\rho'$.
Let $R$ be a register (label) system spanned by $\{\ket{i}\}_{i=1}^{n}$ with a trivial Hamiltonian $H_R=I$, where $I$ is an identity operator.
Then we set catalyst $C$ as $S^{\otimes n-1}\otimes R$ and its state $c$ as
\eq{
c=\frac1n \sum_{k=1}^n \rho^{\otimes k-1}\otimes \Xi'_{n-k}\otimes \dm{k},
}
where $\Xi'_i$ is the reduced state of $\Xi'$ to the first $i$-th copies of $S$, that is, $\Xi'_i:=\Tr_{\overline{S^{\otimes i}}}[\Xi']$.
The initial state of the composite system $S\otimes C$ reads
\eq{
\rho\otimes c=\frac1n \sum_{k=1}^n \rho^{\otimes k}\otimes \Xi'_{n-k}\otimes \dm{k}.
}
Here, the last $n-1$ copies of $S$ are assigned to $C$.
We apply $\Lambda$ on $S\otimes C=S^{\otimes n}\otimes R$ if the label is $\ket{n}$.
Then, the obtained state with proper relabeling indeed fulfills the desired correlated catalytic conversion.

Thus, it suffices to construct the asymptotic marginal conversion from $\rho$ to $\rho'$ with rate arbitrarily close to 1 with arbitrary accuracy.

\bigskip
{\bf Outline of the asymptotic marginal conversion protocol}

We now outline the steps of our asymptotic marginal conversion protocol:
\be{
\item[(0)] We first decompose $n$ copies of $\rho$ into $\nu$ blocks of $\mu$ copies of $\rho$ and an additional sublinear ($o(n)$) copies of $\rho$ (i.e., $n=\mu\nu+o(n)$).
\item We apply a dephasing channel to each $\mu$ copies $\rho^{\otimes \mu}$ so that all the off-diagonal elements vanish.
The state becomes an energy-diagonal state $\rho_{\cl,\mu}\in S^{\otimes  \mu}$.
\item We convert energy-diagonal state $\rho_{\cl,\mu}^{\otimes \nu}$ to another energy-diagonal state $\rho'^{\otimes \nu}_{\cl,\mu}$ with the same number of copies by a thermal operation. Here, $\rho'_{\cl,\mu}$ is a state with the same spectrum as $\rho'^{\otimes  \mu}$ and satisfies $\frac{1}{ \mu}F(\rho'_{\cl,\mu})\simeq \frac{1}{ \mu}F(\rho'^{\otimes  \mu})$ but $F(\rho'_{\cl,\mu})> F(\rho'^{\otimes  \mu})$.
\item In parallel with steps (i) and (ii), we distill a state $\eta$ with broad coherence from $o(n)$ copies of $\rho$.
\item Using $\eta$ repeatedly as catalytic coherence, we convert energy-diagonal state $\rho'_{\cl,\mu}$ to coherent state $\rho'^{\otimes  \mu}$.
}
Below we explain these steps in detail.

\bigskip
{\bf Reduction to energy-diagonal states} 

In our asymptotic marginal conversion protocol, we first decompose $n$ copies of $\rho$ into $\nu$ sets of $\mu$ copies of $\rho$ and sublinear ($o(n)$) copies of $\rho$ (step (0)).
The latter sublinear number of copies are used at step (3).
In step (1), we first apply a dephasing channel known as pinching to each $\mu$ copies, $\rho^{\otimes \mu}$, so that all the off-diagonal elements vanish and the state becomes an energy-diagonal state $\rho_{\cl,\mu}\in S^{\otimes \mu}$.
It is known that the decrease of free energy density through pinching can be arbitrarily small by setting $\mu$ sufficiently large~\cite{Hayashi2002optimal}:
\eq{
\frac1\mu F(\rho^{\otimes \mu})\simeq\frac1\mu F(\rho_{\cl,\mu}).
}
Through step (1), we obtain $\nu$ copies of $\rho_{\cl,\mu}$.

The state is now energy-diagonal, which enables us to disregard all unwanted effects from quantum coherence.
We here introduce the incoherent counterpart of $\mu$ copies of the final state $\rho'^{\otimes \mu}$ denoted by $\rho_{\cl,\mu}'\in S^{\otimes \mu}$.
We require $\rho_{\cl,\mu}'$ to have the same spectrum as $\rho'^{\otimes \mu}$ and to satisfy 
\eq{
\frac1\mu F(\rho_{\cl,\mu}')\simeq \frac1\mu F(\rho'^{\otimes \mu}) \ {\rm  but} \ F(\rho_{\cl,\mu}')> F(\rho'^{\otimes \mu}).
}
Recall that the thermal operation and the Gibbs-preserving operation have the same state convertibility for energy-diagonal states, and the necessary and sufficient condition of state conversions by a Gibbs-preserving operation in the asymptotic regime has already been established as the second law~\cite{Bra13, Fai19}.
Hence, if $F(\rho_{\cl,\mu})\geq F(\rho_{\cl,\mu}')$ holds, there exists a thermal operation converting $\rho_{cl,\mu}^{\otimes \nu}\to \rho_{\cl,\mu}'^{\otimes \nu}$ with a sufficiently large $\nu$.
Through this procedure, we obtain states $\nu$ copies of $\rho_{\cl,\mu}'$ with arbitrary accuracy.

\bigskip
{\bf Catalytic coherence} 

We finally convert incoherent state $\rho_{\cl,\mu}'$ to the desired coherent state $\rho'^{\otimes \mu}$.
To this end, we employ the technique of {\it catalytic coherence}~\cite{Abe14}.
This technique allows us to implement arbitrary unitary operations including coherence-generating channels by using only energy-conserving unitaries.
To this end, we introduce an external system $E$ prepared in a broadly coherent state $\eta$, which absorbs the backaction imposed by energy conservation.
Thanks to the broad coherence of $\eta$, the coherent aspects of unitary operation are faithfully implemented by an energy-conserving operation.

To grasp the idea of catalytic coherence, let us consider as an example the implementation of the Hadamard gate $H=\ketbra{+}{0}+\ketbra{-}{1}$ where $\ket{\pm}=(\ket{0}\pm\ket{1})/\sqrt{2}$.
Let $\Di^{\pm}$ be the shift operator shifting the energy of the external system by $\pm D$, where $D$ is the energy difference from $\ket{0}$ to $\ket{1}$.
We implement the Hadamard gate on $S$ by the following energy-conserving unitary on $SE$:
\eq{
U= \frac{1}{\sqrt{2}}I+\frac{1}{\sqrt{2}}\ketbra{0}{1}\otimes \Di^++\frac{1}{\sqrt{2}}\ketbra{1}{0}\otimes \Di^-. \nt
}
If $\eta$ has broad coherence (e.g., $\ket{\eta}=\frac{1}{\sqrt{n}}\sum_{i=1}^n \ket{i}$), two states $\eta$ and $\Di^{\pm}\eta\Di^{\mp}$ have large overlap, suggesting that the Hadamard gate, which is a coherence-generating channel, is implemented accurately: 
\eq{
\Tr_E[U(\rho\otimes \eta)U^\dagger]\simeq H\rho H^\dagger.
}
Notably, the ability of $\eta$ for this implementation does not decrease through this procedure, and the external system can be reused in the next emulation.
Namely, we can employ $\eta'=\Tr_S[U(\rho\otimes \eta)U^\dagger]$ also as the initial state of the external system $E$ instead of $\eta$.

\bigskip
{\bf Random-walk argument for coherence recovery}

Now we apply this implementation technique to the last step of our asymptotic marginal conversion protocol.
We firstly distill a coherent resource state $\eta$ from the remaining $o(n)$ copies of $\rho$ (step (3)).
Then we implement a unitary operation converting $\rho_{\cl,\mu}'$ into $\rho'^{\otimes \mu}$ on all $\nu$ sets repeatedly by reusing $\eta$ as catalytic coherence (step (4)).
Owing to the reusability of the catalytic coherence, a single copy of $\eta$ suffices to implement the unitary operation by $\nu=O(n)$ times.

However, an additional argument is needed to establish an accurate implementation by the catalytic coherence protocol. 
The problem is that, the distribution in $E$ over the energy levels spreads, and its end will reach the ground state.
We fail to continue repeated implementations if a state touches the ground state, since we cannot lower the energy in this case.
Thus, to demonstrate accurate repeated unitary implementation, we need to show that the probability of this hitting event of the ground state can be made arbitrarily small.

To this end, we map the dynamics of the coherent resource state onto a classical random walk.
For simplicity, suppose that $\mu=1$ and the eigenenergies of the system $S$ are given by $E_n=nE$, with level spacing $E$.
Now, we construct the random walk on the ladder of energy levels corresponding to the unitary $U$ with $U_{ij}=\braket{E_i|U|E_j}$ as follows.
We expand the quasi-classical state after step (2) as $\rho_{\cl,\mu=1}'=\sum_n p'_n\ket{E_m}\bra{E_n}$.
Then, the probability, denoted by $P(c)$, that the random walker (the coherent resource state) moves upward by $c$ levels in a single step is given by
\eq{
P(c)=\sum_{(i,j), j-i=c} \abs{U_{ij}}^2 p'_j.
}
Although the actual quantum state in our quantum composite system $SE$ is a quantum superposition of these dynamics, by tracing out the main system $S$ the dynamics of the coherent resource system is properly described by this classical picture.

Importantly, if $F(\rho_{\cl,\mu=1}')>F(\rho')$ is satisfied, which is fulfilled in the case of free energy ordering $F(\rho)\geq F(\rho')$, this random walk is biased to the upward direction:
\eq{
\sum_c cP(c)>0.
}
Here, in the case of $F(\rho)=F(\rho')$ we slightly modify the state $\rho'$ to $\rho''$ within the permitted error $\ep$ such that $F(\rho)>F(\rho'')$ and rename $\rho''$ as $\rho'$.
This ensures $F(\rho)> F(\rho')$ in the case of the free energy ordering, leading to the upward bias.
It is well known that with high probability a random walk with drift do not hit a site at finite distance against the drift from the initial position.
Adopting this result, we conclude that pumping up the coherent resource state by only finite ($O(1)$) energy level suffices to avoid the hitting event of the ground state with arbitrarily high probability, which guarantees an accurate implementation of the desired unitary operations.

\let\oldaddcontentsline\addcontentsline
\renewcommand{\addcontentsline}[3]{}

\bibliographystyle{apsrmp4-2}
\bibliography{myref}

\let\addcontentsline\oldaddcontentsline

\bigskip

{\bf \normalsize Date Availability}

Data sharing not applicable to this article as no datasets were generated or analysed during the current study.

\bigskip

{\bf \normalsize Acknowledgement}

N.S. acknowledges the support of JST ERATO Grant No. JPMJER2302, Japan.
R.T. acknowledges the support of JST CREST Grant Number JPMJCR23I3, JSPS KAKENHI Grant Number JP24K16975, JP25K00924, and MEXT KAKENHI Grant-in-Aid for Transformative
Research Areas A ``Extreme Universe” Grant Number JP24H00943.

\bigskip

{\bf \normalsize Author contributions}

NS put key ideas to derive main results, and RT refined these ideas.
Both authors contributed to the interpretation
and discussion of the results, as well as to the writing of the manuscript. 

\bigskip

{\bf \normalsize Competing Interest}

The authors declare no competing interest.

}

\newpage
\newgeometry{hmargin=1.2in,vmargin=0.8in}

\onecolumngrid


\newcommand{\vo}{\upsilon}
\newcommand{\midskip}{\vspace{3pt}}

\setcounter{equation}{0}

\makeatletter
\def\shorttableofcontents#1#2{\bgroup\c@tocdepth=#2\@restonecolfalse
  \settowidth\js@tocl@width{\headfont\prechaptername\postchaptername}%
  \settowidth\@tempdima{\headfont\appendixname}%
  \ifdim\js@tocl@width<\@tempdima \setlength\js@tocl@width{\@tempdima}\fi
  \ifdim\js@tocl@width<5zw \divide\js@tocl@width by 5 \advance\js@tocl@width 4zw\fi
\if@tightshtoc
    \parsep\z@
  \fi
  \if@twocolumn\@restonecoltrue\onecolumn\fi
  \@ifundefined{chapter}%
  {\section*{{#1}
        \@mkboth{\uppercase{#1}}{\uppercase{#1}}}}%
  {\chapter*{{#1}
        \@mkboth{\uppercase{#1}}{\uppercase{#1}}}}%
  \@startshorttoc{toc}\if@restonecol\twocolumn\fi\egroup}
\makeatother


\newpage

\begin{center}
{\large \bf Supplemental Material of ``Recovery of the second law in fully quantum thermodynamics''}
 \\
 \vspace*{0.3cm}
 Naoto Shiraishi and Ryuji Takagi  \\
 \vspace*{0.1cm}

 {Department of Basic Science, The University of Tokyo} 

\end{center}

\tableofcontents

\section{Preliminaries}

\subsection{Thermal and covariant operations}

A central goal in quantum thermodynamics is to characterize feasible state transformation with respect to the property of the initial and target states.  
Here, we are particularly interested in transforming the initial state $\rho$ to another target state $\rho'$, both of which are in the $d$-dimensional system $S$ with Hamiltonian $H_S$.
In this manuscript, we generally write the Hamiltonian of a system $X$ as $H_X$.
Without loss of generality, we assume that $H_S$ has the ground energy 0 and order their eigenenergies $\{E_i\}_{i=1}^{d}$ in non-decreasing order $0= E_1\leq E_2\dots \leq E_d$ with potential degeneracy. 
We let $\ket{E_i}$ denote an energy eigenstate with energy $E_i$.
If $E_i=E_j$, we take $\ket{E_i}$ and $\ket{E_j}$ to be orthogonal energy eigenstates with the same energy.
We assume that the system is surrounded by a heat bath with a given fixed inverse temperature $\beta$.
We write the thermal Gibbs state for the system $X$ by $\tau_{\Gibbs,X} = e^{\beta H_X}/\Tr(e^{-\beta H_X})$.
We suppress the subscript $X$ if the relevant system is clear from the context.

Given this setting, a reasonable choice as the set of accessible thermodynamic operations is the operations such that the system may interact with the thermal bath while respecting the energy conservation law, where the thermal bath is finally discarded. 
This is the set of thermal operations~\cite{HO13,Sagawa2021asymptotic-reversibility,Fai19}.

\begin{dfn}[Thermal operations~\cite{Sagawa2021asymptotic-reversibility}]
A channel $\Lambda$ from $S$ to $S'$ is called a thermal operation if there is an environmental system $E$ with Hamiltonian $H_E$, the Gibbs state $\tau_{\Gibbs,E} = e^{-\beta H_E}/\Tr(e^{-\beta H_E})$, and energy-conserving partial isometry $V_{SE\to S'E}$ with $V_{SE\to S'E}(H_S+H_E)=(H_{S'}+H_E)V_{SE\to S'E}$ satisfying 
\bal
 \Lambda(\rho) = \Tr_E\left[V_{SE\to S'E}(\rho\otimes \tau_{\Gibbs,E}) V^\dagger_{SE\to S'E}\right].
\eal
\end{dfn}

There has been much research that investigated the operational capability of thermal operations in state transformation. 
However, most of the results are restricted to the limited settings where the target states are block-diagonal in the energy eigenbasis, corresponding to a {\it semi-classical} scenario. 
This is due to the technical difficulty in characterizing the manipulation of \emph{quantum coherence}---the off-diagonal parts of density matrices---which represents a specific quantum feature.

The main contribution of our results is to comprehensively deal with the manipulation of quantum coherence and find the effective use of it to establish the fully quantum second law of quantum thermodynamics. 
To discuss coherence manipulation in a quantitative manner, it is central to consider another class of operations that cannot create or increase quantum coherence, known as \emph{covariant operations}~\cite{Bartlett2007reference,Gour2008resource,Mar-thesis}.

\begin{dfn}[Covariant operations]
 A channel $\Lambda$ from $S$ to $S'$ is called a covariant operation if 
\bal
 \Lambda(e^{-iH_S t}\rho e^{iH_S t}) = e^{-iH_{S'}t}\Lambda(\rho) e^{iH_{S'} t}\quad
\eal    
is satisfied for all $t$.
\end{dfn}

It is straightforward to check that thermal operations are always covariant. 
This means that thermal operations themselves cannot create quantum coherence, which is why quantum coherence behaves as a precious resource under thermodynamic transformation.

In general, the framework that investigates the feasible state transformation under the given set of accessible operations is known as quantum resource theories~\cite{Chitambar2019quantum}.
The set of operations accessible in the given physical settings is often called the set of free operations and is denoted by $\bbO$, which is a subset of quantum channels (completely-positive trace-preserving maps). 
The framework of quantum thermodynamics introduced above can be considered as a resource theory with $\bbO$ being the set of thermal operations~\cite{Bra13}. 
If we focus on coherence manipulation without taking into account the non-equilibrium energy distribution, the relevant resource theory is known as a resource theory of coherence or asymmetry~\cite{Gour2008resource}, where the covariant operations are chosen as its free operations.

\subsection{Free energy as an information-theoretic quantity}

As our main theorem claims, the non-equilibrium free energy of the quantum state $\rho$ defined by 
\bal
 F(\rho) = \Tr(\rho H) - S(\rho)/\beta
\eal
serves as a fundamental quantity that reflects the resourcefulness in the thermodynamics scenario. 
An important property of this quantity is that it can never increase under free thermodynamic processes. 
This property can be clearly seen by rewriting it as 
\bal
 F(\rho) = D(\rho\|\tau_{\Gibbs})/\beta + F(\tau_{\Gibbs})
\eal
where $F(\tau_{\Gibbs}) = -\log \Tr(e^{-\beta H})/\beta$ is the equilibrium free energy, and 
\bal
 D(\rho\|\sigma) = \Tr(\rho\log\rho) - \Tr(\rho\log\sigma)
\eal
is the quantum relative entropy for arbitrary states $\rho$ and $\sigma$.
This ensures that for a thermal operation $\Lambda:S\to S$, we get for any state $\rho$ on $S$ that 
\bal
 F(\Lambda(\rho))-F(\tau_{\Gibbs}) = D(\Lambda(\rho)\|\tau_{\Gibbs})/\beta = D(\Lambda(\rho)\|\Lambda(\tau_{\Gibbs}))/\beta\leq D(\rho\|\tau_{\Gibbs})/\beta = F(\rho)- F(\tau_{\Gibbs}), 
\eal
implying $F(\Lambda(\rho))\leq F(\rho)$.
Here, in the second equality we used that a thermal operation maps the Gibbs state to Gibbs state, and the inequality is due to the data-processing inequality of quantum relative entropy. 

This formulation of the free energy in terms of quantum relative entropy allows us to use information-theoretic tools to analyze this quantity. 
Since $D(\rho\|\tau_{\Gibbs})$ is equivalent to free energy up to constant factor and off-set by the equilibrium free energy, here we sometimes call $D(\rho\|\tau_{\Gibbs})$ simply free energy.

\subsection{Catalytic and asymptotic transformations}\label{sec:catalytic and asymptotic}

The setting of state transformation focused on this work is the one where we use an ancillary state to aid the state transformation, while the aiding state is returned to its original form. 
Such a transformation is known as a catalytic transformation. 
Depending on in what form the original catalyst state is returned, there are several potential settings of catalytic transformation. 
Our main focus in this article is on the one known as a {\it correlated catalyst}, where the final state may have a correlation between the main system and the catalytic system, while the reduced state on the catalytic system has exactly the same form as the initial catalyst state. 

\begin{dfn}[Correlated-catalytic transformation]\label{dfn:correlated catalytic}
We say that a state $\rho$ can be transformed to a state $\rho'$ with error $\epsilon$ by a set $\bbO$ of free operations with a correlated catalyst if there exists a state $c$ on a finite-dimensional system $C$ and a quantum channel $\Lambda\in\bbO$ such that 
\bal
 \Lambda(\rho\otimes c) = \pi_{SC},\quad \|\Tr_C(\pi_{SC})-\rho'\|_1\leq \epsilon,\quad \Tr_S(\pi_{SC})=c.
\eal
    
\end{dfn}

Besides the catalytic framework, we also have a useful framework of state transformations, an asymptotic transformation, which considers transformation between an infinite number of copies of states and studies the optimal rate of the number of initial and final states.  

The standard asymptotic transformation aims to transform $\rho^{\otimes n}$ to $\rho'^{\otimes rn}$ so that the error in state transformation approaches 0 in the limit of $n\to \infty$. The figure of merit in this setting is the maximum achievable rate $r$. 
Formally, the asymptotic transformation rate from $\rho$ to $\rho'$ by a set $\bbO$ of free operations is defined as 
\bal
 R(\rho\to\rho')\coloneqq \sup\lset r\sbar \exists \Lambda_n\in\bbO\ \ {\rm s.t.} \  \lim_{n\to\infty}\|\Lambda_n(\rho^{\otimes n})-\rho'^{\otimes \lfloor rn\rfloor}\|_1=0\rset.
 \label{eq:standard asymptotic rate}
\eal

For our later use, we here introduce another (arguably less familiar) type of asymptotic transformation known as asymptotic marginal transformation~\cite{Ferrari2023asymptotic,GKS24}. 
For the states $\rho$ and $\rho'$ on the system $S$, we consider transforming $\rho^{\otimes n}$ on $S^{\otimes n}$ to a state $\rho'_{\lfloor rn\rfloor}$ on the system $S^{\otimes \lfloor rn \rfloor}$ so that each reduced state approaches to the target state $\rho'$ at $n\to\infty$ limit. 
Formally, we define the asymptotic marginal transformation rate under a set $\bbO$ of free operations by 
\bal
 R_{\rm m}(\rho\to\rho')\coloneqq \sup\lset r\sbar \exists \Lambda_n\in\bbO,\ \Lambda_n:S^{\otimes n}\to S'^{\otimes \lfloor rn \rfloor},\ \ {\rm s.t.}\ \forall i, \  \lim_{n\to\infty}\|\Tr_{\overline{S}_i}\Lambda_n(\rho^{\otimes n})-\rho'\|_1=0\rset
 \label{eq:aymsptotic marginal}
\eal
where $\Tr_{\overline{X}}$ denotes the partial trace taken over all the systems except $X$.
We particularly say that $\rho$ can be transformed to $\rho'$ by an asymptotic marginal transformation if $R_{\rm m}(\rho\to\rho')=1$.

As we will see in the later section, our construction first performs the asymptotic marginal asymptotic transformation and converts it to the \emph{single-shot} correlated-catalytic transformation (i.e., a single copy of $\rho$ to a single copy of $\rho'$) as formalized in Definition~\ref{dfn:correlated catalytic}.

\subsection{Quantum coherence and coherent modes}

In this work, quantum coherence refers to the superposition between energy eigenstates with different eigenenergies. 
In the density matrix representation, it corresponds to the off-diagonal terms with respect to the energy eigenbasis corresponding to different energy levels.  
To formalize it, we first introduce the quantum channel on system $X$ known as the pinching channel defined by 
\bal
 \calP(\rho) = \lim_{T\to\infty}\frac{1}{T}\int_{-\infty}^\infty e^{-iH_Xt}\rho e^{iH_X t}  = \sum_E \Pi_E \rho \Pi_E
 \label{eq:pinching}
\eal
where $\Pi_E$ is the projector onto the subspace with eigenenergy $E$. 
Namely, the pinching channel is a decohering channel that removes the quantum coherence, i.e., off-diagonal part corresponding to the superposition between eigenstates with different eigenenergies. 
We note that the pinching channel is a thermal operation~\cite{Gour2024resourcesquantumworld,Watanabe2025universalworkextractionquantum}.
With this, we say that a state $\rho$ is \emph{incoherent} or \emph{classical} if $\rho = \calP(\rho)$, i.e., it is block-diagonal with respect to energy eigenbasis. 
Otherwise, $\rho$ is called a coherent state. 

It turns out that we need more fine-grained notion of quantum coherence to characterize the feasible state transformation.
Coherence is a superposition between two eigenstates with different energies, and it is crucial to take into account the value of the energy difference. 
This notion was introduced in Ref.~\cite{MS14} as modes of asymmetry, which we call coherent modes here to be consistent with other terminology in this article.

\begin{dfn}[Coherent modes~\cite{MS14}] \label{def:mode of asymmetry}
Coherent modes of a state $\rho$ are the energy differences for which $\rho$ has nonzero coherence.
\begin{equation}
\calD(\rho):=\lset \Di_{ij} \sbar \Di_{ij}=E_i-E_j,\ (i,j)\in \calM(\rho)\rset,
\end{equation}
where $\calM(\rho)$ is the set of integer pairs $(i,j)$ such that there exists a pair $\ket{E_i}$, $\ket{E_j}$ of energy eigenstates with energies $E_i$ and $E_j$ such that $\braket{E_i|\rho|E_j}\neq 0$.
\end{dfn}

The coherent modes monotonically decrease under a covariant operation.

\blm{[Monotonicity of coherent modes~\cite{MS14}]
Suppose that there exists a covariant operation $\Lambda$ transforming $\rho$ to $\rho'=\Lambda(\rho)$.
Then, we have $\calD(\rho')\subseteq \calD(\rho)$.
}

In contrast to the above monotonicity, under the presence of catalysts, linear combinations of coherent modes can also be created under covariant operations~\cite{TS22,ST23}. 
Ref.~\cite{ST23} particularly showed that the integer-linear combination of coherent modes can be created with the help of correlated catalysts and thus played a major role there to characterize the possible catalytic transformation. 
This naturally motivates us to focus the set of linear combinations of coherent modes with integer coefficients and investigate its role in characterizing the state transformation in our setting.

\begin{dfn}[Resonant coherent modes~\cite{ST23}] \label{dfn:rationally resonant coherent modes}
We define the set $\calC(\rho)$ of integer resonant coherent modes of a state $\rho$ as
\begin{equation}
\calC(\rho):=\lset x  \sbar x=\sum_{\Di\in \calD(\rho)} a_\Di \Di,\ a_\Di\in \bbZ\rset,
\end{equation}
where $\bbZ$ is the set of integers. 
\end{dfn}

For later use, for any set of real numbers $\calS$ we introduce a set of numbers obtained by a linear combination of elements in $\calS$ with integer coefficients as 
\eq{
\calI(\calS):=\lset x  \sbar x=\sum_{s\in \calS} a_s s,\ a_s\in \bbZ\rset.
}
Using this symbol, the resonant coherent modes are also given by $\calC(\rho)=\calI(\calD(\rho))$.

\section{Asymptotic transformation (First part of the proof of the main theorem)}

The fact that $F(\rho)\geq F(\rho')$ is necessary for a correlated-catalytic thermal operation to transform $\rho$ to $\rho'$ with arbitrary accuracy follows from the superadditivity of free energy~\cite{Wilming2017axiomatic,TS22}. 
Therefore, we focus on the achievable part, i.e., if $F(\rho)\geq F(\rho')$ and $\calC(\rho')\subset\calC(\rho)$, then there exists a correalted-catalytic thermal operation that transforms $\rho$ to $\rho'$ with arbitrary accuracy.

Here, if $F(\rho)=F(\rho')$, for any given error $\ep>0$ in the correlated-catalytic transformation, we can find $\rho''$ such that $\| \rho''-\rho'\|_1<\ep$ and $F(\rho)>F(\rho'')$.
Then, by redefining $\ep$ as $\ep/2$ and $\rho'$ as $\rho''$, without loss of generality we suppose $F(\rho)>F(\rho')$ without equality.
In the following, we show the desired correlated-catalytic transformation under this condition $F(\rho)>F(\rho')$.

\bigskip

Our construction to show the main theorem consists of two parts. 
In the first part, we focus on realizing the \emph{asymptotic transformation} from $\rho$ to $\rho'$.
We particularly provide a protocol that transforms $n$ copies of $\rho$ to $n-o(n)$ copies of $\rho'$ in the sense of asymptotic marginal transformation. 
In the second step, we employ this asymptotic transformation to construct a single-shot transformation from $\rho$ to $\rho'$ with a correlated catalyst. 

In this section, we focus on the first part concerning the asymptotic transformation, and the asymptotic-catalytic correspondence is argued in the next section.
Our goal of this section is to prove the following theorem.

\bthm{
Consider two quantum states $\rho$ and $\rho'$ on $S$ satisfying $F(\rho)\geq F(\rho')$ and $\calC(\rho')\subset\calC(\rho)$.
Then, $R_{\rm m}(\rho\to\rho')=1$ holds.
Furthermore, the final correlation between a subsystem and the rest can be made arbitrarily small.
}

\subsection{Setting up the stage: ladder systems}\label{sec:ladder}

Let us first introduce the concept of linear independence relevant in our settings. 

\begin{dfn}
  Let $\calS=\{x_i\}_i$ be a set of real numbers. We call $\calS$ integer-linearly independent if $\sum_i a_i x_i = 0$ with $a_i\in\bbZ$ implies $a_i=0$ for all $i$.
\end{dfn}

The above definition ensures that if $\calS=\{x_i\}_i$ is integer-linearly independent and if a real number $y$ admits the form $y=\sum_j a_j x_j$ with $a_j\in\bbZ$, then there is a one-to-one correspondence between $y$ and the set $\{a_i\}_i$ of integers.

The next result shows that one can always construct an integer-linearly independent set as a basis for a given set of real numbers, admitting an integer-linear decomposition for every element in the given set.

\begin{lm}\label{lm:rational linearly independent}
  Let $\calT=\{y_i\}_i$ be a set of real numbers. Then, there exists an integer-linearly independent set $\calS=\{x_i\}_i$ of real numbers such that $\calI(\calS)=\calI(\calT)$.
\end{lm}
\begin{proof}
  We explicitly construct such a set as follows. 
  Let us start with the empty set $\calS=\{\}$ and sequentially update $\calS$ by the following procedure. 
  \begin{enumerate}
      \item Suppose that we have $\calS=\{x_i\}_i$ at the $k$\,th step. 
      We see whether $y_k\in\calT$ can be written as an integer-linear combination of the current $\calS$. 
      If so, we do not add $y_k$ and move to the next element $y_{k+1}$.
      If not, we then examine whether $\calS\cup\{y_k\}$ is an integer-linearly independent set.
      \begin{enumerate}
          \item If so, add $y_k$ in $\calS$.
          \item If not, we have a relation $0=\sum_{i\in I} a_i x_i$ for some subset $I$ of labels where $a_i$'s are nonzero integers and $x_i$'s are elements of $\calS\cup\{y_k\}$.
          Without loss of generality, we assume that $a_i$'s do not have a common factor.
          By assumption, we have relatively prime two integers $a_i$ and $a_j$ with $\abs{a_i}<\abs{a_j}$, which allows us to express $a_j= pa_i+q$ with $1\leq \abs{q}<\abs{a_i}$.
          We then replace as $x_i\to x_i+px_j$ and keep other variables.
          Remark that a set of integers obtained by a linear combination of $\{x_l\}_l$ with integer coefficients is invariant under this replacement.
          In addition, we again have the relation $0=\sum_i a_i x_i$ with replacing as $a_j\to q$.
          We continue this replacement until some coefficient $a_l$ to be one.
          At this time, we safely set $\{x_i\}_{i(\neq l)}$ as $\calS$ since $x_l$ can be expressed by a linear combination of $\{x_i\}_{i(\neq l)}$ with integer coefficients.          
      \end{enumerate}
      \item We repeat this process until all $y_k\in\calT$ are considered. 
  \end{enumerate}
The resulting set $\calS$  then serves as the desired basis set.
\end{proof}

We construct the set $\calS=\{\Di_j\}_j$ from the set $\{E_i\}_{i=1}^d$ of energy levels of $H_S$ and consider $|\calS|$ ladder systems $\otimes_{j=1}^{|\calS|}L_j$ where $L_j$ has the Hamiltonian $H_{L_j}=\sum_{k=-K_j}^{K_j} \Di_j k\dm{k}$ for some $K_j$.
Here, we take a finite but sufficiently large $K_j$ so that the later operations are feasible.
Then, each energy eigenstate of $\otimes_{j=1}^{|\calS|}L_j$ can be denoted by $\ket{E_{\bf c}} = \otimes_{j=1}^{|\calS|}\ket{c_j}$, which has the energy $E_{\bf c}=\sum_j c_j\Di_j$.

For analytical convenience, we first embed the system $S$ to the ladder systems $\otimes_{j=1}^{|\calS|}L_j$.
This mapping can be done by a thermal operation because $V_{S\to \otimes_{j=1}^{|\calS|}L_j} = \sum_{\bf c} {}_{\otimes_j L_j}\ketbra{E_{\bf c}}{E_{\bf c}}_S$ is an energy conserving isometry.
Recall that $H_S$ may be degenerate and have $E_i=E_j$ for $i\neq j$, while $\braket{E_i|E_j}=0$.
To accommodate the degeneracy, we assume that the ladder system also hosts degenerate energy eigenstates.
This allows us to focus on the protocol acting on the ladder system.
We can then come back to the original system at the end of the protocol by applying the partial isometry $V_{S\to \otimes_{j=1}^{|\calS|}L_j}^\dagger$.
Therefore, from now on, we often call the ladder system simply $S$ and focus on working on the ladder system.

Separately from $\calS$, we also construct another integer-linearly independent set $\calQ=\{\tilde\Di_l\}_l$ of real numbers by applying the protocol in Lemma~\ref{lm:rational linearly independent} to the coherent mode $\calD(\rho)$ of input states. 
Noting that $\calI(\calQ)=\calI(\calC(\rho))$, this constructs a ``basis'' of the set $\calC(\rho)$ of resonant coherent modes. 
Based on $\calQ$, we also consider a ladder system $Q=\otimes_{q=1}^{|\calQ|} Q_l$ where $H_{Q_l}=\sum_{k=-\tilde K_l}^{\infty}\tilde\Di_l k\dm{k}$ for some sufficiently large integer $\tilde K_l$.


\subsection{Overview of the protocol}

We aim to construct a thermal operation that, for an arbitrary given real number $\epsilon>0$, transforms $\rho^{\otimes n}$ on $S^{\otimes n}$ to a state $\Xi_{S^{\otimes n'}}$ on $S^{\otimes n'}$ with $n'=n-o(n)$ such that $\|\Tr_{\overline{i}}\Xi_{S^{\otimes n'}} - \rho'\|_1\leq \epsilon$ for all $i$. 
Here is the overview of the protocol. (See also Fig.~\ref{f:schematic})

\figin{8.7cm}{schematic2}{
{\bf Overview of the asymptotic marginal conversion protocol (same as \fref{schematic} in the main article).}
We transform $n$ copies of $\rho$ into $n'=n-o(n)$ copies of $\rho'$ in the sense of the asymptotic marginal conversion.
As its preliminary treatment, we first divide $n$ copies of the initial state $\rho$ into $n'=\mu\nu$ copies and $o(n)$ copies.
In step (1) we use the dephasing channel and remove all the off-diagonal elements of each $\mu$ copies of $\rho$. 
The state becomes an energy-diagonal state $\rho_{\cl,\mu}$ on $S^{\otimes \mu}$.
Through this process, the decrease of the free energy density is arbitrarily suppressed.
In step (2), we convert $\nu$ copies of $\rho_{\cl,\mu}$ into $\nu$ copies of another energy-diagonal state $\rho_{\cl,\mu}'$, which has the same spectrum as $\rho'^{\otimes \mu}$.
The existence of this conversion is confirmed by the second law for classical states.
In step (3), we distill state $\eta$ with broad coherence from $o(n)$ copies of $\rho$.
In step (4), using $\eta$ as catalytic coherence repeatedly, we recover $\rho'^{\otimes \mu}$ from $\rho_{\cl,\mu}'$ by an energy-conserving unitary operation.
The accuracy of this implementation is confirmed by the random-walk argument.
}{schematic_supplemental}

\begin{enumerate}

    \item \label{item:pinching} 
    We write $n'=\mu\nu$ for integers $\mu$, $\nu$, which we choose later large enough to achieve the desired accuracy in the transformation. We first see $\rho^{\otimes n'}$ as $\nu$ groups of $\rho^{\otimes \mu}$ and separately transform each $\rho^{\otimes \mu}$ into a classical state (energy eigenstates). We do this by applying the pinching channel $\calP$ in Eq.~\eqref{eq:pinching} to $\rho^{\otimes \mu}$, obtaining $\rho_{\cl,\mu}\coloneqq\calP(\rho^{\otimes \mu})$.
    By taking sufficiently large $\mu$, the free energy per single copy $\frac{1}{\mu}D(\calP(\rho^{\otimes \mu})\|\tau_{\Gibbs}^{\otimes \mu})$ can be made arbitrarily close to the original free energy $\frac{1}{\mu}D(\rho^{\otimes \mu}\|\tau_{\Gibbs}^{\otimes \mu})=D(\rho\|\tau_{\Gibbs})$.

    \item \label{item:transforming classical states}
    We transform classical state ${\rho}_{\cl,\mu}^{\otimes \nu} = \calP(\rho^{\otimes \mu})^{\otimes \nu}$ to another classical state $\rho_{\cl,\mu}'^{\otimes \nu}$, where $\rho_{\cl,\mu}'$ is a classical state such that its eigenvalues are the same as $\rho'^{\otimes \mu}$ and its eigenstates have similar energies to $\rho'^{\otimes \mu}$. 
    Previous studies~\cite{Bra13} have shown that the condition $D(\rho_{\cl,\mu}\|\tau_{\Gibbs}^{\otimes \mu})\geq D(\rho'_{\cl,\mu}\|\tau_{\Gibbs}^{\otimes \mu})$ for classical states $\rho_{\cl,\mu}$ and $\rho'_{\cl,\mu}$ guarantees the presence of the desired asymptotic transformation (in the standard sense of Eq.~\eqref{eq:standard asymptotic rate}) by a thermal operation.
    Fortunately, we indeed have this condition since
    \bal
    \frac{1}{\mu}D(\rho_{\cl,\mu}\|\tau_{\Gibbs}^{\otimes \mu})\sim D(\rho\|\tau_{\Gibbs})\geq D(\rho'\|\tau_{\Gibbs})\sim \frac{1}{\mu}D(\rho'_{\cl,\mu}\|\tau_{\Gibbs}^{\otimes \mu})
    \eal
    holds and two approximations become arbitrarily accurate by setting sufficiently large $\mu$.

\item \label{item:resource state}  
To apply a coherent unitary in the next step, we distill coherent resource states $\ket{\eta_{L}}=\frac{1}{\sqrt{L}}\sum_{i=0}^L \ket{i}$ from $\rho^{\otimes o(n)}$ on half-infinite ladder systems $Q_l$ with energy interval $\Di_l\in\calQ$.
We choose $L$ so that the desired accuracy is achieved in the next step.
We will see that for our purpose we can set a sufficiently large $L$ to a number depending only on $\mu$, not on $\nu$. 
Hence, by taking $\mu=O(1)$ and $\nu=O(n)$, it suffices to use $O(1)$ copies of $\rho$ to get $\ket{\eta_L}$ working in the next step.

\item \label{item:Aberg} 
We transform $\rho_{\cl,\mu}'^{\otimes \nu}$ to $\rho'^{\otimes \mu\nu}$. 
We achieve this by using the energy-shifted coherent resource states $\ket{\eta_{L,M}}\coloneqq \frac{1}{\sqrt{L}}\sum_{i=0}^L\ket{i+M}$ on $Q_l$ for all $l$ with pumping energy units $M$ into the resource state obtained in the previous step.
We implement a coherent unitary by an energy-conserving unitary together with the coherent resource state and apply this coherent unitary to one copy of $\rho_{\cl,\mu}'$ that rotates the energy eigenbasis of $\rho'_{\cl,\mu}$ to the eigenbasis of $\rho'^{\otimes \mu}$. 
Although the coherent resource states on $Q_l$ change from the original states through this implementation, we reuse these states as coherent resource states and implement the same coherent unitary by $\nu$ times.
We will show that if the free energy ordering $F(\rho)\geq F(\rho')$ is satisfied, then the error of these repeated implementations can be bounded arbitrarily small by taking $M$ sufficiently large but $O(1)$ (independent of $n$).
\end{enumerate}

Besides the $n$ copies of the initial state $\rho$ that are to be transformed to the target state, we use additional $o(n)$ copies of $\rho$ in Steps~(\ref{item:resource state})~and~(\ref{item:Aberg}).
However, those only add $o(n)$ (in fact, constant) number of copies and thus do not affect the ratio, resulting in the unit rate $\frac{n-o(n)}{n}\xrightarrow[n\to\infty]{} 1$.

In the following, we discuss each step in detail.

\subsection{Step~(\ref{item:pinching}): Transform initial states to classical states}

We here transform $\rho^{\otimes \mu}$ to a classical state by applying the pinching channel \eqref{eq:pinching} to $\rho^{\otimes \mu}$ to get 
\bal
\rho_{\cl,\mu}\coloneqq\calP(\rho^{\otimes \mu}),
\eal
which is a classical mixture of energy eigenstates (classical states). 
Since the pinching channel removes coherence, the free energy is generally reduced accordingly during this process. 
Nevertheless, it is known that the decrease in the free energy per copy converges to zero in the infinite copy limit~\cite{Hayashi2002optimal, Gou22} 
\bal
  \lim_{\mu\to\infty} \frac{1}{\mu}D(\calP(\rho^{\otimes \mu})\| \tau_{\Gibbs}^{\otimes \mu}) = D(\rho\|\tau_{\Gibbs}). 
  \label{eq:pinched free energy}
\eal
Hence, by setting $\mu$ sufficiently large, this dephasing process only reduces the free energy per copy by an arbitrarily small amount.

We separately apply $\calP$ to $\nu$ groups of $\rho^{\otimes \mu}$ to obtain $\calP(\rho^{\otimes \mu})^{\otimes \nu}=\rho_{\cl,\mu}^{\otimes \nu}$.

\subsection{Step~(\ref{item:transforming classical states}): Transform classical states}

We express $\mu$ copies of the target state $\rho'$ via the spectral decomposition:
\bal
\rho'^{\otimes \mu}=\sum_j \lambda_j \dm{\psi_j}
\eal
with eigenvalues $\{\lambda_j\}_j$ and eigenbasis $\{\ket{\psi_j}\}_j$. 
We consider the corresponding classical state $\rho'_{\cl,\mu}$ of the form
\bal
 \rho'_{\cl,\mu} = \sum_j \lambda_j \dm{E_{{\bf c}[j]}}.
\eal
Here $\ket{E_{{\bf c}[j]}}$ is an energy eigenstate with energy $E_{{\bf c}[j]} = \sum_l ({\bf c}[j])_l\Di_l$, where ${\bf c}[j]=(c_1[j],\dots, c_{|\calS|}[j])$ is the vector chosen as follows.

Let 
\bal
\ket{\psi_j} = \sum_{\bf c} f_{j,{\bf c}}\ket{E_{\bf c}}
\eal
for some coefficients $f_{j,{\bf c}}$.
Let $\bsc'$ be a vector satisfying $f_{j,\bsc'}\neq 0$.
Then, we choose $\bsc[j]$ so that 
\bal
E_{\bsc[j]}-E_{\bsc'}\in \calC(\rho),
\label{eq:coherent mode condition}
\eal
which guarantees the desired unitary transformation.
Eq.~\eqref{eq:coherent mode condition} allows us to write 
\bal
 E_{{\bf c}[j]}-E_{\bsc'} = \sum_l m_{j\bsc',l} \tilde \Di_l
 \label{eq:decomposing energy difference into ladder interval}
\eal
where $m_{j\bsc',l}$ is an integer and $\tilde\Di_l$ is an element of $\calQ$, i.e., the basis of coherent resonant modes (See Sec.~\ref{sec:ladder}).
We require that the energy $E_{\bsc[j]}$ has a slightly larger energy unit for every mode in $\calQ$ than $\ket{\psi_j}$ on average, i.e.,
\bal
  \sum_{\bsc'} m_{j\bsc',l}|f_{j\bsc'}|^2 > 0 \geq  \sum_{\bsc'} m_{j\bsc',l}|f_{j\bsc'}|^2 - 1,\quad \forall l,j
  \label{eq:classical state energy condition general}.
\eal

We note in passing that the condition \eqref{eq:coherent mode condition} ensures that in the next step we can extract necessary coherent modes from $\rho$.
Precisely, we need to make sure that the unitary that rotates $\ket{E_{{\bf c}[j]}}$ back to $\ket{\psi_j}$---which we discuss in the next step---is something that we can implement using the coherent modes extractable from $\rho$, which satisfies the condition $\calC(\rho')\subseteq\calC(\rho)$.
In other words, the condition $E_{\bsc'}-E_{\bsc[j]}\in \calC(\rho)$ should be satisfied for all $\bsc'$ with $f_{j,\bsc'}\neq 0$.
We touch on this issue in more detail in Step~(\ref{item:Aberg}).

The role of condition \eqref{eq:classical state energy condition general} is two-fold. The first inequality ensures that the average energy only decreases in the system during the basis rotation process in Step~(\ref{item:Aberg}), which allows us to use the same coherent resource state repeatedly. 
On the other hand, the second inequality guarantees that we can transform the classical state $\rho_{\cl,\mu}$ obtained in Step~(\ref{item:pinching}) to $\rho_{\cl,\mu}'$ asymptotically as follows.

We argue that one can transform $\rho_{\cl,\mu}^{\otimes \nu}$ into $\rho_{\cl,\mu}'^{\otimes \nu}$ by a thermal operation with an arbitrary accuracy by taking a sufficiently large $\nu$. 
We first note that
\bal
 D(\rho'_{\cl,\mu}\|\tau_{\Gibbs}^{\otimes \mu}) &= \beta\sum_j \lambda_j E_{{\bf c}[j]}  - S(\rho'_{\cl,\mu}) \\
 & = \beta\sum_j \lambda_j E_{{\bf c}[j]}  - S(\rho'^{\otimes \mu})\\
 & = \beta\sum_j \lambda_j \left(\braket{\psi_j|H|\psi_j}+\sum_{\bsc',l} |f_{j\bsc'}|^2 m_{j\bsc',l}\tilde\Di_l\right) - S(\rho'^{\otimes \mu})\\
 & < \beta\sum_j \lambda_j \left(\braket{\psi_j|H|\psi_j}+\sum_{l=1}^{|\calQ|} \tilde\Di_l\right) - S(\rho'^{\otimes \mu})\\
 & = \mu D(\rho'\|\tau_{\Gibbs}) + \beta \sum_{l=1}^{|\calQ|} \tilde\Di_l.
 \label{eq:relative entropy classical and target}
\eal

Thus we have 
\bal
 \frac{D(\rho_{\cl,\mu}\|\tau_{\Gibbs}^{\otimes \mu})}{\mu}-\frac{D(\rho'_{\cl,\mu}\|\tau_{\Gibbs}^{\otimes \mu})}{\mu} \geq D(\rho\|\tau_{\Gibbs})-D(\rho'\|\tau_{\Gibbs}) -\epsilon_1-\frac{\beta\sum_{l=1}^{|\calQ|} \tilde\Di_l}{\mu} \lb{step2-confirm}
\eal
where $\epsilon_1$ is the deviation of $\frac{D(\rho_{\cl,\mu}\|\tau_{\Gibbs}^{\otimes \mu})}{\mu}=\frac{D(\calP(\rho^{\otimes\mu})\|\tau_{\Gibbs}^{\otimes \mu})}{\mu}$ from $D(\rho\|\tau_{\Gibbs})$.
As argued in the previous subsection, particularly shown in \eqref{eq:pinched free energy}, $\epsilon_1$ can be made arbitrarily small by taking sufficiently large $\mu$.
In addition, as discussed at the beginning of this section, we suppose $D(\rho\|\tau_{\Gibbs})>D(\rho'\|\tau_{\Gibbs})$ without loss of generality.
Since $\sum_{l=1}^{|\calQ|} \tilde\Di_l$ is independent of $\mu$, we now find that there exists a sufficiently large $\mu$ such that the right-hand side of \eref{step2-confirm} becomes positive, particularly guaranteeing $D(\rho_{\cl,\mu}\|\tau_{\Gibbs}^{\otimes \mu})>D(\rho'_{\cl,\mu}\|\tau_{\Gibbs}^{\otimes \mu})$. 

We now recall that, for two classical states $\rho$ and $\sigma$, the asymptotic transformation with vanishing error is possible if and only if $D(\rho\|\tau_{\Gibbs})\geq D(\sigma\|\tau_{\Gibbs})$~\cite{Bra13}.  
This means that the classical state $\rho_{\cl,\mu}^{\otimes \nu}=\calP(\rho^{\otimes \mu})^{\otimes \nu}$ can be transformed to $\rho_{\cl,\mu}'^{\otimes \nu}$ by a thermal operation with an arbitrary accuracy by taking a sufficiently large $\nu$.

\subsection{Step~(\ref{item:resource state}): Distill coherent resource state}

In Step~(\ref{item:Aberg}), we will implement a coherent unitary that transforms $\rho_{\cl,\mu}'$ to the target state $\rho'^{\otimes \mu}$ with thermal operations. 
Since thermal operations are covariant and cannot create coherence from scratch, we need a coherent resource state, with which a thermal operation can implement an (approximate) coherent unitary. 
Here, we aim to prepare the coherent resource state by using a sublinear number of copies of $\rho$.

We first show that a covariant operation can be implemented by a thermal operation with an arbitrary accuracy with the help of non-thermal states. 
\begin{lm}\label{lm:covariant by thermal}
Let $\Phi$ be a covariant channel from $S$ to $S'$, and $\kappa$ be a non-thermal state on another system $\tilde S$.
Then, for an arbitrary real number $\epsilon>0$, there exists a thermal operation $\calE$ and an integer $m$ such that
\bal
 \|\calE(\cdot\otimes\kappa^{\otimes m})-\Phi(\cdot)\|_\diamond\leq\epsilon.
\eal
\end{lm}
\begin{proof}
As shown in Refs.~\cite{Keyl1999optimal,Marvian2008building}, any covariant channel $\Phi:S\to S'$ with $\dim S, \dim S'<\infty$ can be implemented by an energy-conserving unitary $U$ applying on $S$ and some external system $E$ with $\dim H_E<\infty$ equipped with some Hamiltonian $H_E$ as
\bal
 \calE(\rho) = \Tr_{E'} \left[U \left(\rho \otimes \dm{0}_E\right)U^\dagger\right]
 \label{eq:covariant dilution}
\eal
where $\ket{0}_E$ is the zero energy eigenstate of $H_E$.
Therefore, it suffices to show that the pure state $\ket{0}_E$ can be prepared by many copies of non-thermal state $\kappa$.
This can be seen by the fact shown in Ref.~\cite{Bra13} that asymptotic transformation from $\kappa\otimes\tau_{\Gibbs,E}$ to $\tau_{\Gibbs,\tilde S}\otimes\dm{0}_E$ can be done at the rate of
\bal
r&=D(\kappa\otimes \tau_{\Gibbs,E}\|\tau_{\Gibbs,\tilde S}\otimes \tau_{\Gibbs,E})/D(\tau_{\Gibbs,\tilde S}\otimes \dm{0}_E\|\tau_{\Gibbs,\tilde S}\otimes\tau_{\Gibbs,E})\\
&= D(\kappa\|\tau_{\Gibbs,\tilde S})/D(\dm{0}_E\|\tau_{\Gibbs,E})\\
&>0.
\eal
This, together with the fact that partial trace is also a thermal operation, means that there exists a sufficiently large integer $m$ such that $\kappa^{\otimes m}$ can be transformed to $\dm{0}_E^{\otimes rm}$ with an arbitrary accuracy by a thermal operation, which then results in $\dm{0}_E$ by tracing out unwanted copies. 

\end{proof}
This allows us to focus on preparing the desired coherent state with a covariant operation, as in our protocol, we are allowed to use $o(n)$ copies of $\rho$, which serves as the non-thermal state that enables the implementation of covariant operations.

Then, the following result shows that one can always distill the maximally coherent state with a specific coherence mode from the given state.

\begin{lm}\label{lem:coherence distillation}
    Let $\calD(\rho)$ be the coherent modes of $\rho$ in Definition~\ref{def:mode of asymmetry}, and let $\ket{+}_\delta=\frac{1}{\sqrt{2}}(\ket{0}+\ket{1})$ be the maximally coherent bit on the two-level system with Hamiltonian $H=\delta\dm{1}$. 
    Then, for an arbitrary energy level $\delta\in\calD(\rho)$ in the coherent modes, an arbitrary integer $k>0$, and an arbitrary real number $\epsilon>0$, there exists an integer $\nu$ and a thermal operation $\Lambda$ such that 
    \bal
 \|\Lambda(\rho^{\otimes \nu}) - \dm{+}_\delta^{\otimes k}\|_1<\epsilon.
\eal
\end{lm}
\begin{proof}

The result in Ref.~\cite[Lemma S.16]{ST23} ensures that for an arbitrary coherence mode $\delta\in\calD(\rho)$, an arbitrary integer $k>0$, and a real number $\epsilon'>0$, there exists an integer $\nu'$ and a covariant operation $\Phi$ such that 
\bal
 \|\Phi(\rho^{\otimes \nu'}) - \dm{+}_\delta^{\otimes k}\|_1<\epsilon'.
 \label{eq:coherence distillation}
\eal
This covariant operation $\Phi$ can be implemented by a thermal operation by using additional copies of $\rho$ as ensured by Lemma~\ref{lm:covariant by thermal}.

\end{proof}

Furthermore, the copies of a two-level coherent state can be transformed to a uniform coherent state on a higher-dimensional system. 

\begin{lm}\label{lm:dilution}
Let $\ket{+}_\delta=\frac{1}{\sqrt{2}}(\ket{0}+\ket{1})$ be the maximally coherent state on the two-level system $T$ with Hamiltonian $H_T=\delta \dm{1}$, and let $\ket{\eta_L}\coloneqq \frac{1}{\sqrt{L}}\sum_{j=0}^{L-1} \ket{j}$ be the uniform coherent state over $L$ energy levels on a finite-dimensional system $T'$ with Hamiltonian $H_{T'}=\sum_{k=-K}^{K} k\delta\dm{k}$ for a positive integer $K>L$.
For an arbitrary $\epsilon>0$ and a positive integer $L$, there exists a sufficiently large integer $m$ such that there is a covariant operation $\Phi$ satisfying 
\bal
 \|\Phi(\dm{+}_\delta^{\otimes m}) - \dm{\eta_L}\|_1\leq \epsilon .
\eal

\end{lm}
\begin{proof}
The result in Ref.~\cite[Theorem 1]{Marvian2022operational-interpretation} ensures that $\dm{+}_\delta$ can be asymptotically transformed to any state on the system $T'$ by a covariant operation. This particularly means that using a sufficiently large number $\dm{+}_\delta$, a covariant operation can transform them to a single copy of $\ket{\eta_L}$ with an arbitrary accuracy.

\end{proof}

Combining Lemmas~\ref{lm:covariant by thermal}, \ref{lem:coherence distillation}, and \ref{lm:dilution}, we obtain the following. 

\begin{lm}\label{thm:prepareation of coherent resource state}
For an arbitrary $\epsilon>0$, positive integer $L$, and a state $\rho$ such that $\Di\in\calD(\rho)$, there exists a sufficiently large integer $m$ and a thermal operation $\Lambda$ such that 
\bal
 \|\Lambda(\rho^{\otimes \mu}) - \dm{\eta_L}_{\Di}\|_1\leq \epsilon .
\eal
\end{lm}
\begin{proof}
 Lemma~\ref{lem:coherence distillation} ensures that one can transform copies of $\rho$ into the $\ket{+}_\Di$ with an arbitrary accuracy by a covariant operation, where $\ket{+}_\Di = \frac{1}{\sqrt{2}}(\ket{0}+\ket{1})$ is the maximally coherent state on the two-level system with Hamiltonian $H_{\Di}=\Di\dm{1}$.  
 Lemma~\ref{lm:dilution} then ensures that one can transform copies of $\ket{+}_{\Di}$ to $\ket{\eta_L}_\Di$.
 Finally, Lemma~\ref{lm:covariant by thermal} ensures that this covariant operation can be implemented by a thermal operation by using the additional number of copies of $\rho$.
 
\end{proof}

By running this distillation process for every $\tilde\Di_l\in\calD(\rho)$, we get the coherent resource state $\ket{\eta_L}_{\tilde\Di_l}$ for every $\tilde\Di_l\in\calD(\rho)$.
This eventually provides the coherent state 
\bal
\eta_{\calD(\rho)}=\otimes_{l:\tilde\Di_l\in\calD(\rho)}\dm{\eta_L}_{\tilde\Di_l}
\label{eq:coherent resource state on coherent modes}
\eal
with arbitrary accuracy. 
We label the system that hosts $\ket{\eta_L}_{\tilde\Di_l}$ as $\tilde Q_l$, which has a non-degenerate energy levels with Hamiltonian $H_{\tilde Q_l}=\sum_{j=-K}^K j\tilde\Di_l \dm{j}$ for an integer $K>L$.

We now argue that, using these states, we can prepare $\ket{\eta_L}_{\tilde\Di}$ on the system $Q_{\tilde\Di}$ with Hamiltonian $H_{\tilde\Di}=\sum_j j\tilde\Di\dm{j}$ for an arbitrary $\tilde\Di\in\calQ$.
The crucial point is that, because $\calQ\subset\calC(\rho)$, any $\tilde\Di\in\calQ$ can be written by an integer-linear combination of the elements in $\calD(\rho)$ as $\tilde\Di=\sum_l c_l\tilde\Di_l$ 
for some integers $c_l$ and $\tilde\Di_l\in\calD(\rho)$.
Then, one can implement an arbitrary unitary $V$ on $Q_{\tilde\Di}$ with arbitrary accuracy by applying~\cite{Kitaev2004superselection,Abe14,TS22,ST23} 
\bal
 U_{Q_{\tilde\Di} \tilde Q}=\sum_{j,{\bf c, c'}} \dm{j'} V \dm{j}_{\Delta} \otimes \bigotimes_l S^{c_l(j-j')}_{\tilde Q_l} + W_{Q_{\tilde\Di} \tilde Q}
 \label{eq:coherence transfer}
\eal
with $\eta_{\calD(\rho)}$ on $\otimes_l \tilde Q_l$ with a sufficiently large $L$, where $S_{\tilde Q_l}\coloneqq \sum_{j=0}^\infty \ketbra{j+1}{j}_{\tilde\Di_l}$ is the ladder-shift operator, and $W_{Q_{\tilde\Di} \tilde Q}$ is an energy-conserving operator that makes $U_{Q_{\tilde\Di} \tilde Q}$ an energy-conserving unitary.
In particular, this allows us to prepare $\ket{\eta_L}_{\Di}$ on $Q_\Delta$ with arbitrary accuracy by implementing a unitary $V$ such that $V\ket{0}=\ket{\eta_L}_{\Di}$, which creates the desired coherent resource state from the incoherent pure state on $Q_{\tilde\Di}$, which can be prepared with arbitrary accuracy by using a sufficient (constant) number of copies of $\rho$ by a thermal operation following the same argument in the proof of Lemma~\ref{lm:covariant by thermal}.
Calling the system that hosts $\ket{\eta_L}_{\tilde\Di_l}$ for $\tilde\Di_l\in\calQ$ $Q_l$, we showed that the coherent resource state 
\bal
\eta_{\calQ}=\otimes_{l:\tilde\Di_l\in\calQ}\dm{\eta_L}_{\tilde\Di_l}
\label{eq:coherent resource state}
\eal
with an arbitrary constant $L$ can be prepared with arbitrary accuracy by using a constant number of the input state $\rho$.

\subsection{Step~(\ref{item:Aberg}): Transform classical states to the target state}

We finally transform the state $\rho_{\cl,\mu}$ to $\rho'^{\otimes \mu}$ by using the coherent resource state obtained in Step~(\ref{item:resource state}).
This amounts to implementing repeatedly the unitary $U_\mu$ that connects $\ket{E_{{\bf c}[j]}}$ and $\ket{\psi_j}$ as $U_\mu\ket{E_{{\bf c}[j]}}=\ket{\psi_j}$.
In this subsection, we first show that we can successfully implement the unitary $U_\mu$, and then we demonstrate that we can repeatedly apply this implementation by reusing a single coherent resource state accurately.

We first clarify the action of the unitary $U_\mu$, which works as 
\bal
 U_\mu\rho_{\cl,\mu}' U_\mu^\dagger = \sum_j \lambda_j U_\mu\dm{E_{{\bf c}[j]}} U_\mu^\dagger = \sum_j \lambda_j \dm{\psi_j} = \rho'^{\otimes \mu}. 
 \label{eq:unitary basis rotation}
\eal
As $\ket{\psi_j}$ is generally not an energy eigenstate, the unitary $U_\mu$ needs to be a coherence-generating unitary.  
We implement this by using a coherent resource state. 
Recalling that $\ket{\psi_j}=\sum_{\bf c}f_{j,{\bf c}}\ket{E_{\bf c}}$, we particularly consider the following energy-conserving unitary applying over $S^{\otimes \mu}$ and the system $\otimes_l Q_l$ given by 
\bal
 U_{SQ}=\sum_{j,{\bf c, c'}} f_{j,{\bf c'}}\ketbra{E_{\bf c'}}{E_{{\bf c}[j]}}_S\otimes \bigotimes_l S^{m_{j\bsc',l}}_{Q_l} + W_{SQ}
 \label{eq:Aberg unitary}
\eal
where $m_{j\bsc',l}$ is the coefficient in \eqref{eq:decomposing energy difference into ladder interval}, and $W_{SQ}$ is an energy-conserving operator that makes $U_{SQ}$ unitary.
This unitary operator is constructed so that the energy change in the system $S$ is compensated by the energy change in the ladder system $Q$.
We implement the desired unitary $U_\mu$ by employing $U_{SQ}$ with the resource state $\eta_{\calQ}$ in \eqref{eq:coherent resource state} by having a sufficiently large $L$~\cite{Abe14,TS22} in order to achieve an arbitrary accuracy.

Recalling that $\calC(\rho)$ can be constructed by integer linear combination of $\calQ$, we conclude that the above implementation indeed succeeds if $\calC(\rho)$ contains all the relevant modes in creating $\ket{\psi_j}$ from $\ket{E_{\bsc'[j]}}$, and it is precisely ensured by the condition \eqref{eq:coherent mode condition}.

At the first glance, it appears to require us to implement $U_\mu^{\otimes \nu}$ in order to create $\rho'^{\otimes \mu\nu}$.
However, since we take $\mu=O(1)$ with respect to the initial number $n$ of copies, implying $\nu=O(n)$, we need to implement $U_\mu^{\otimes \nu}$ by a resource state $\ket{\eta_{L}}$ with exponentially large $L$ with respect to $n$~\cite{Abe14} or prepare $\nu=O(n)$ copies of resource states $\ket{\eta_{L}}$ with $L=O(1)$, both of which cannot be prepared by using a sublinear number of copies of $\rho$.
This drastically affects the asymptotic transformation rate and will not lead to our main claim.

Fortunately, we can overcome this difficulty by modifying the target state in a way that it still meets our needs.
Instead of preparing $\rho'^{\otimes \mu\nu}$ with arbitrary accuracy, we aim to obtain a state $\Xi$ on the system $S^{\otimes \mu\nu}$ such that a marginal state on every subsystem $S$ is close to $\rho'$, i.e.,  $\Tr_{\overline{i}}[\Xi]\simeq \rho'$ for all $i$. 
To this end, we \emph{reuse} the same coherent resource state $\nu$ times, which requires us to prepare the resource state $\ket{\eta_{L}}$ only once.
Indeed, Ref.~\cite{Abe14} showed that after application of the unitary $U_{SE}$ in \eqref{eq:Aberg unitary}, the marginal state on $Q$ can be reused as a resource state in the next round. 
The crucial observation there was that, although the resource state itself gets changed, it realizes the implementation of the unitary on $S$ with exactly the same accuracy---with no degradation in the capability of unitary implementation. 
This allows us to keep using the coherent resource state that ends up on system $E$ after the previous implementation of $U_{\mu}$---in particular, we aim to apply $U_{\mu}$ sequentially for $\nu=O(n)$ times so that we get the desired state $\Xi$ that locally approximates $\rho'$ with an arbitrary accuracy (Fig.~\ref{fig:Aberg_repeat}). 

\begin{figure}
    \centering
    \includegraphics[width=0.7\linewidth]{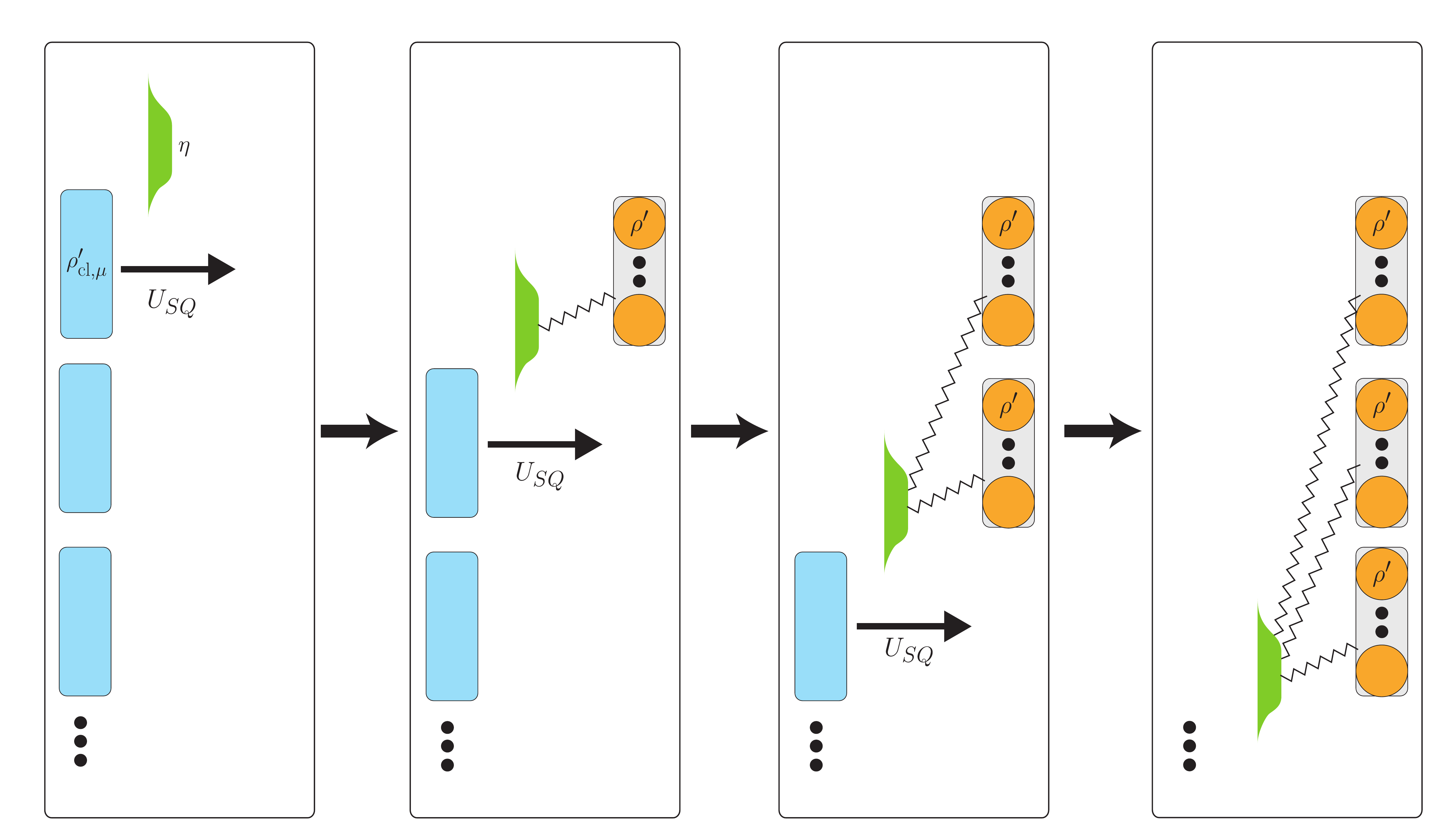}
    \caption{Sequential implementation of the coherent unitary $U_\mu$ with a coherent resource state $\eta$. This approximately produces a state whose marginal on each subsystem $S$ is close to $\rho'$ with potential correlation with other subsystems and catalytic system.}
    \label{fig:Aberg_repeat}
\end{figure}

\figin{8cm}{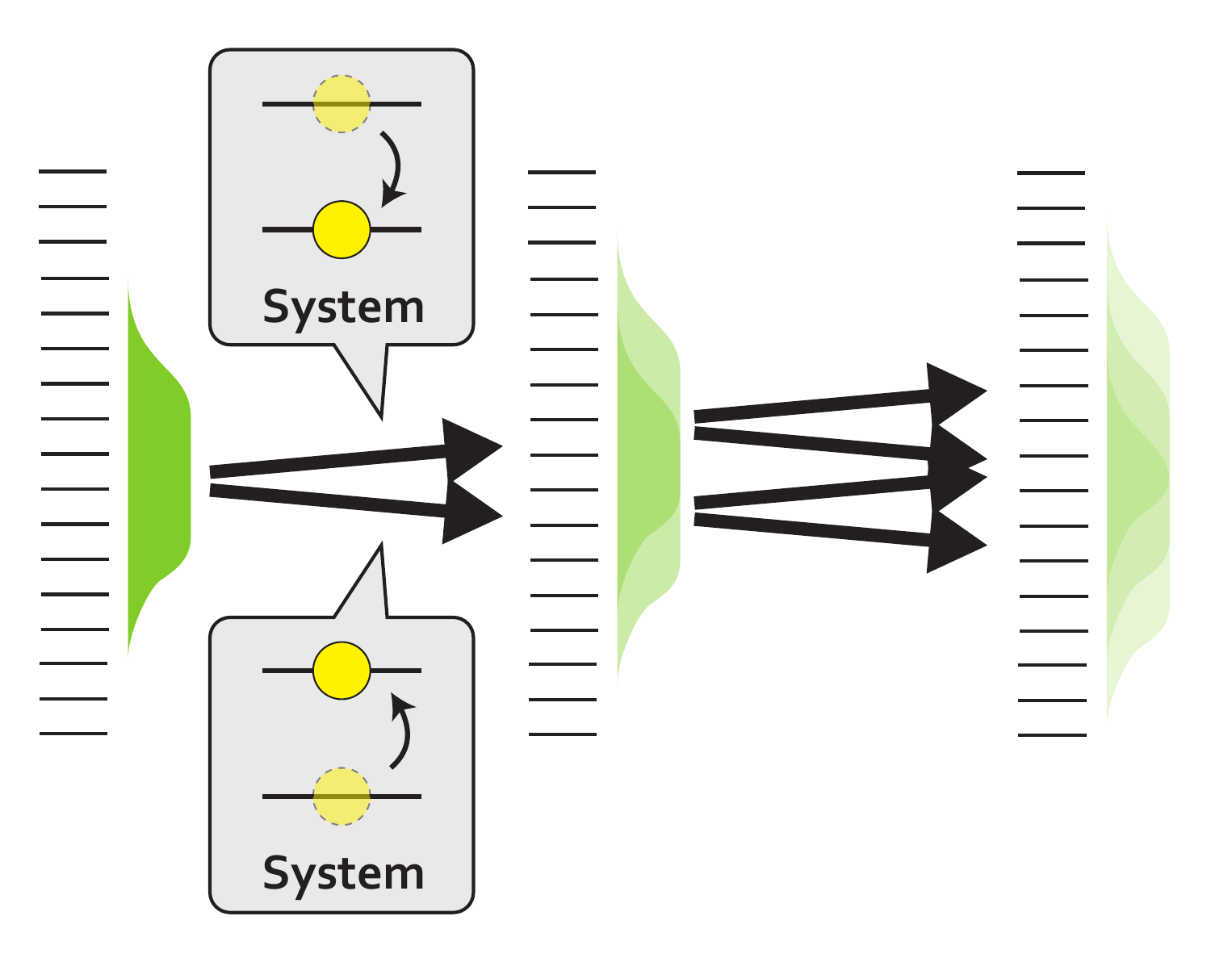}{
The dynamics of the resource state on $Q$.
If the energy in the system is lowered (resp. raised), the resource state is raised (resp. lowered).
This dynamics can be evaluated by using the theory of classical random walks.
}{random-walk}

Here we need a careful examination on the validity of repeated applications. 
After one application of the unitary implementation, the resource state on $Q$ itself changes in a way that the energy levels involved in the state are spread out (\fref{random-walk}). 
This particularly becomes a problem when this spread hits the ground state---once the energy spread is blocked at the ground energy level, the perfect repeatability of this protocol does not stand any longer and the error in implementing $U_\mu$ starts increasing from that point.

We shall demonstrate that the above undesirable event has only negligible probability, which is shown by modeling this energy spread as a random walk. 
As we will discuss, the constraint \eqref{eq:classical state energy condition general} ensures that the random walk moves to the positive direction, allowing us to bound the probability of hitting the ground energy using the martingale theory.
This analysis particularly guarantees that supplying a sufficiently large but \emph{constant} energy $M=O(1)$ for the energy-shifted resource state
\bal
\ket{\eta_{L,M}}\coloneqq\frac{1}{\sqrt{L}}\sum_{i=0}^{L-1} \ket{i+M}
\eal
suffices to keep the performance of repetitive use of coherent states in the desired accuracy. 

We here note that, if the resource state hits the \emph{upper} limit of the Hamiltonian, it would cause a similar issue. 
This can be avoided by considering an infinite-dimensional system without an upper bound of energy levels---the ground energy needs to exist from the physical requirement, but the upper bound does not need to exist. 
Therefore, we first embed the resource state $\eta_{\calQ}$ in \eqref{eq:coherent resource state} on system $Q=\otimes_l Q_l$ with $H_{Q_l}=\sum_{j=1}^L j\tilde\Di_l\dm{j}$ prepared in Step~(\ref{item:resource state}) for some finite $d$ into the infinite-dimensional system  $\tilde Q=\otimes_l \tilde Q_l$ with the half-ladder Hamiltonian $H_{\tilde Q_l}=\sum_{j=0}^\infty j\tilde\Di_l \dm{j}$, which can be done by an energy-conserving isometry $V_{Q_l\to\tilde Q_l}=\sum_{j=1}^L {}_{\tilde Q_l}\dm{j}_{Q_l}$ and thus thermal operation.
In the following, we simply call this infinite half-ladder system $Q$ and assume that the resource state is already embedded in this system. 

This allows us to focus on not hitting the ground state during the repetitive implementation of coherent unitary. 
We first argue that we can pump up the energy of the resource state by a thermal operation together with an additional supply of the non-thermal input state.

\begin{lm}\label{lm:energy pumping}
  Let $Q$ be the system with the Hamiltonian $H_Q=\sum_{j=0}^\infty j\Di \dm{j}$, and $\eta$ be an arbitrary state on $Q$.
  Let $\rho$ be a non-thermal state on another system $S$. 
  Then, there exists an integer $\lambda$ such that for any $\epsilon>0$, there exists an integer $M$ and a thermal operation $\Lambda$ satisfying 
  \bal
   \|\Lambda(\eta\otimes \rho^{\otimes \lambda M}) - \Delta^M_Q\eta\Delta^{\dagger M}_Q \|_1<\epsilon,
  \eal
where $S_Q=\sum_{j=0}^\infty\ketbra{j+1}{j}$ is the shift operator.
\end{lm}

\begin{proof}
This lemma directly follows from Lemma \ref{lm:covariant by thermal}, since the shift operation is covariant.
\end{proof}

Lemma~\ref{lm:energy pumping} allows us to prepare a coherent resource state with shifted energy levels. 
We aim to set this energy shift large enough so that in the repeated application of coherent unitary, the energies involved in the resource state do not hit the ground energy.
In the following, we argue that a sufficiently large but constant ($O(1)$ with respect to $n$) energy shift guarantees the desired accuracy in implementing the coherent unitary repeatedly.
The key idea in the analysis is to see the change in the resource state as a random walk and show that the hitting probability after the \emph{infinite} repetition can be bounded arbitrarily small.

Let us consider $U_{SQ}$ in \eqref{eq:Aberg unitary}.
If we use the coherent resource state $\eta$ on $Q$, the remaining coherent state after the application of $U_{SE}$ is 
\bal
 \Tr_S U_{SQ}(\rho_{\cl,\mu}\otimes\eta) U_{SQ}^\dagger = \sum_{j,{\bf c'}} \lambda_j |f_{j,{\bf c'}}|^2\left(\otimes_l\Delta^{m_{j\bsc',l} }_{Q_l}\right)\eta\left( \otimes_l{\Delta^\dagger}^{m_{j\bsc',l} }_{Q_l}\right)
 \label{eq:coherent state after rotation}
\eal
This implies that the dynamics of $\eta$ can be described as the random walk, which occurs on each ladder system $Q_l$ with transition probability $p_{l,\delta} = \sum_j \sum_{{\bf c'}:m_{j\bsc',l}=\delta}\lambda_j |f_{j{\bf c}'}|^2$. 
Moreover, the average move of each random walk is positive, which can be seen as follows.
Since $\sum_{\bsc'}m_{j\bsc',l}|f_{j\bsc'}|^2\geq 0$ because of \eqref{eq:classical state energy condition general}, we get 
\bal
 \sum_\delta \delta p_{l,\delta} = \sum_\delta\sum_j \sum_{{\bf c'}:m_{j\bsc',l}=\delta}\delta\lambda_j |f_{j{\bf c}'}|^2 = \sum_{j,{\bf c'}} m_{j\bsc',l}\lambda_j |f_{j{\bf c'}}|^2 \geq 0.
\eal
This means that the average move of $\eta$ on the $l$\,th ladder is positive.

This argument shows that the energy shift in the coherent resource state can be modeled by a random walk with moving toward positive direction on average. 
This in turn implies that, if we start the random walk from a sufficiently large positive energy, the probability for the walk to hit the ground energy can be well bounded. 
The following result shows that even a sufficiently large constant energy shift independent of $n$ is sufficient to achieve an arbitrary given accuracy.

\begin{lm} \label{lm:random-walk-hit}
Consider a discrete-time one-dimensional classical random walk on integers $\bbN$.
Let $X_n$ be a stochastic variable at the $n$-th step, which represents the position.
Suppose that the jump probability depends only on the difference between two positions: $\bbP(X_{n+1}-X_n=i)=p_i$ with given $p_i$'s, where $\bbP$ denotes the probability.
The jump is not long, i.e., we have a fixed $l$ such that $p_i=0$ for $\abs{i}>l$.

Suppose that this random walk has a positive bias: $\sum_{i=-l}^l ip_i>0$.
We set $X_0=0$ with probability 1.
Then, the probability that $X_n$ takes a value less than or equal to $-\xi$ ($\xi>0$) is bounded from above as
\eqa{
\bbP(\exists n \ {\rm s.t.} \ X_n\leq -\xi)\leq \frac{1}{\gamma^{-\xi-1}-\gamma^{-1}+1}
}{P-bound-RW}
Here, $\gamma<1$ is a solution of
\bal
f(\gamma):=\sum_{i=1}^l \( \sum_{j=1}^i \gamma^j\) p_i - \sum_{i=1}^l \( \sum_{j=1}^i \gamma^{-j}\) p_{-i}=0.
\label{eq:gamma-solution}
\eal

\end{lm}

Note that \eqref{eq:gamma-solution} has a solution less than 1 since $f(\gamma)$ is a strictly increasing function of $\gamma$ and $f(1)=\sum_{i=-l}^l ip_i>0$.
Since the right-hand side of \eqref{P-bound-RW} is further bounded as $\frac{1}{\gamma^{-\xi-1}-\gamma^{-1}+1}<\gamma^\xi$, this relation means an exponential decay of the arrival probability.

\begin{proof}
We introduce a new stochastic variable $Y_n$, which is a function of $X_n$ as
\eq{
Y_n=\bcases{
\sum_{i=1}^{X_n} \gamma^i & X_n>0 \\
0 & X_n=0 \\
-\sum_{i=1}^{-X_n} \gamma^{-i} & X_n<0.
}
}
This map is given as $Y_n(X=i)-Y_n(X=i-1)=\gamma^i$.
By construction of $\gamma$ in \eqref{eq:gamma-solution}, $Y_n$ is martingale.

We consider a stopping process that if $X_n\leq -\xi$, then we stop this process.
For an integer $a$, define $T_a^{\downarrow}=\inf\lset n\sbar X_n\leq a\rset$ and $T_a^{\uparrow}=\inf\lset n\sbar X_n\geq a\rset$, and let $T\coloneqq\min\{T_{-\xi}^\downarrow,T_\mu^\uparrow\}$ for a positive integer $\mu$. 
Since $\min\{T,n\}$ is a bounded stopping time, we can apply the optional stopping theorem~\cite{Grimmett2020probability} to get 
\bal
 \bbE[Y_{\min\{T,n\}}] = \bbE[Y_0] = 0,
\eal
where $\bbE$ denotes the expected value.
Together with that $\min\{n,T\}\to T$ with $n\to\infty$ almost surely, and $\lim_{n\to\infty}\bbE[Y_{\min\{T,n\}}]=\bbE[Y_T]$, we get 
\bal
 \bbE[Y_T] = 0.
 \label{eq:expectaion stopping time}
\eal

Since $P(T<\infty)=1$, we have $P(T=T_{-\xi}^\downarrow)+P(T=T_\mu^\uparrow)=1$.
This, together with \eqref{eq:expectaion stopping time}, implies
\bal
 0 = P(T=T_{-\xi}^\downarrow) Y_{T_{-\xi}^\downarrow} + P(T=T_\mu^\uparrow) Y_{T_\mu^\uparrow}.
 \label{eq:martingale average}
\eal
We observe that 
\bal
 Y_{T_{-\xi}^\downarrow} \leq -\sum_{i=1}^{\xi}\gamma^{-i} = -\frac{\gamma^{-\xi}-1}{1-\gamma},\qquad
 Y_{T_{\mu}^\uparrow} \xrightarrow[\mu\to\infty]{} \sum_{i=1}^{\infty} \gamma^{i}&=\frac{\gamma}{1-\gamma}.
 \label{eq:martingale variable bounds}
 \eal

Letting $P\coloneqq \bbP(\exists n \ {\rm s.t.} \ X_n\leq -\xi)$, we note that $P=\bbP(T=T_{-\xi}^\downarrow)$ when taking $\mu\to\infty$.
Therefore, \eqref{eq:martingale average} and \eqref{eq:martingale variable bounds} imply 
\bal
 -\frac{\gamma^{-\xi}-1}{1-\gamma} P + \frac{\gamma}{1-\gamma}(1-P)\geq 0,
\eal
which results in 
\bal
 P\leq \frac{1}{\gamma^{-\xi-1}-\gamma^{-1}+1}.
\eal

\end{proof}

Lemma~\ref{lm:random-walk-hit} ensures that, by taking $M$ sufficiently large depending on the target accuracy, the repeated implementation of coherent unitary $U_\mu$ retains the performance. 
More precisely, suppose that $\rho_{\cl,\mu}^{\otimes \nu}$ are on $S_1^{\otimes \mu}\otimes S_2^{\otimes \mu}\dots \otimes S_\nu^{\otimes \mu}$ where $S_1,\dots, S_\nu$ are identical copies of the system $S$, but we just put the labels to make it explicit that the $i$\,th copy of $\rho_{\cl,\mu}$ acts on $S_i^{\otimes \mu}$.
Let $\calU_{S_i^{\otimes \mu}Q}(\cdot) = U_{S_i^{\otimes \mu}Q}\cdot U_{S_i^{\otimes \mu}Q}^\dagger$ 
be the unitary channel that implements $U_{\mu}$ via \eqref{eq:Aberg unitary}, $P$ be the hitting probability in Lemma~\ref{lm:random-walk-hit} with $\xi=M$, and $\epsilon>0$ be the error that comes with the one implementation of $U_\mu$.
Since we have $|\calQ|$ random walks each of which has no-hitting probability $1-P$, the probability of the successful implementation with error $\epsilon$ happen can be greater than $(1-P)^{|\calQ|}$.
Taking the other failure events into account as the error, we have for all $i=1,\dots,\nu$ that
\bal
 \|\Tr_{\overline{S_i^{\otimes \mu}}Q}\circ \calU_{S_\nu^{\otimes \mu} Q}\circ\dots\calU_{S_2^{\otimes \mu}Q}\circ\calU_{S_1^{\otimes \mu}Q}(\rho_{\cl,\mu}^{\otimes \nu}\otimes \dm{\eta_{L,M}})-\rho'^{\otimes \mu}\|_1 &\leq  (1-P)^{|\calQ|}\epsilon + 1-(1-P)^{|\calQ|}\\
 &\leq \epsilon + |\calQ|\gamma^M
\eal
where the second inequality is due to Lemma~\ref{lm:random-walk-hit}.
Since $\gamma<1$, one can take $M$ sufficiently large to realize an arbitrary desired accuracy.
This concludes the implementation of the desired asymptotic marginal transformation.

We finally comment on the smallness of correlations between a subsystem and the rest.
As shown in Fig.~\ref{fig:Aberg_repeat}, the correlation stems from the reuse of the coherent resource states.
However, fortunately, we implement a unitary operation mapping a pure state $\ket{E_{\bsc[j]}}$ to another pure state $\ket{\psi_j}$ in parallel, and the resulting state of a subsystem is the mixture of pure states obtained through this unitary operation.
Given that a pure state does not correlate with other systems, we conclude that correlation between a subsystem and the rest can be made arbitrarily small by preparing the unitary operation with sufficiently high accuracy.

\section{Single-shot correlated catalytic transformation from asymptotic transformation (Second part of the proof of the main theorem)}

We here prove the reduction from a asymptotic marginal transformation with transformation rate 1 to a single-shot correlated-catalytic transformation.
Although the general method of converting the asymptotic transformation to a single-shot correlated catalytic transformation is well known~\cite{Duan2005multiple-copy,SS21,ST23,fang2024surpassingfundamentallimitsdistillation}, we need to slightly modify the protocol to apply to our setting because our asymptotic transformation involves input and output states of different numbers of copies.
We address this by appropriately placing thermal Gibbs states in our catalyst as discussed in Ref.~\cite{fang2024surpassingfundamentallimitsdistillation}.

The precise definition and statement is as follows.

\begin{lm}
Suppose that $\rho$ can be transformed to $\rho'$ by asymptotic marginal transformation under thermal operations with unit rate, i.e., $R_{\rm m}(\rho\to\rho') =1$. 
Then, there exists a single-shot correlated-catalytic transformation from $\rho$ to $\rho'$ with a vanishing error.
Moreover, if the asymptotic marginal transformation creates arbitrarily small correlation between a subsystem and the rest, then the correlated-catalytic transformation creates arbitrarily small correlation between the system and the catalyst.
\end{lm}

\begin{figure}
    \centering
    \includegraphics[width=\linewidth]{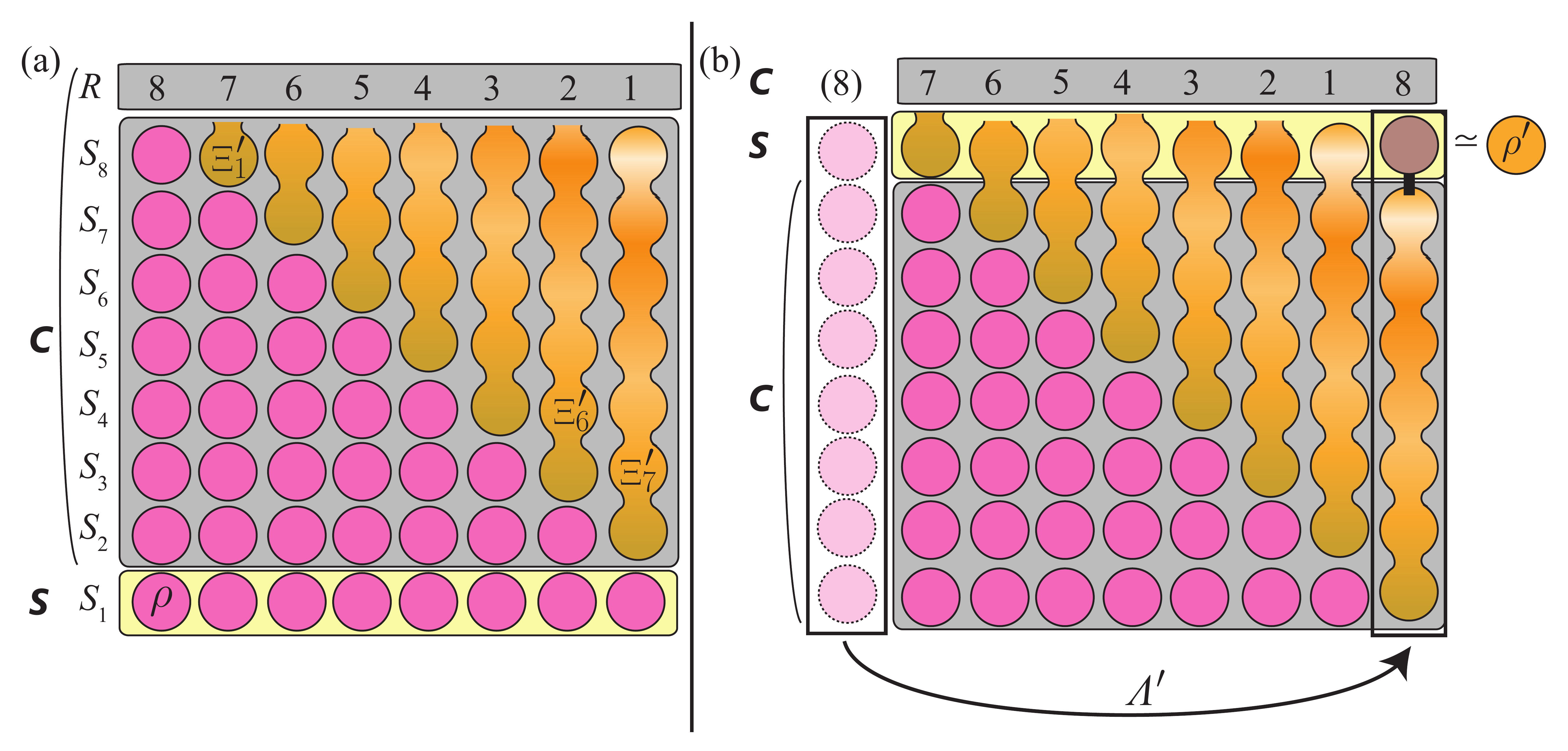}
    \caption{{\bf Construction of catalyst and correlated-catalytic conversion from asymptotic marginal conversion.}
{\bf (a)} We present the initial state of the composite system of the system $S$ and the catalyst $C=S^{\otimes 7}\otimes R$.
The vertical axis represents different copies of systems (i.e., eight different copies of $S$), and the horizontal axis implies different labels (i.e., classical mixture).
The connected circles and a star colored by orange and yellow represents a reduced state of the final state of the asymptotic marginal and correlated-catalytic conversion.
The gray square represents the label system.
{\bf (b)} We apply the asymptotic marginal conversion $\Lambda'$ to the state whose label system is in state 8.
By rearranging the labels of copies of $S$, the system $S$ is at $\rho'$ and the catalyst $c$ returns to its original state. The output of $\Lambda'$ contains a $o(n)$ number of the thermal Gibbs state (one copy in the present example) due to the mismatch of the size of the input and output systems for the asymptotic marginal transformation, which is denoted by a gray circle on the top right.}
    \label{fig:catalyst supplemental}
\end{figure}

\begin{proof}
The assumption $R_{\rm m}(\rho\to\rho')$ means that, for any $\delta>0$ and $\ep>0$, there exists a sufficiently large $n$, a thermal operation $\Lambda:S^{\otimes n}\to S^{\otimes \left\lfloor(1-\delta)n\right\rfloor}$ such that $\|\Tr_{\overline{S_i}}\Lambda(\rho^{\otimes n})-\rho'\|_1\leq \ep$ for every subsystem $S_i$ of $S^{\otimes \left\lfloor(1-\delta)n\right\rfloor}$. 
We modify this thermal operation $\Lambda$ that comes with input and output systems of different sizes to the one with the same input and output systems by attaching the copies of thermal Gibbs state $\tau_{\Gibbs}^{\otimes n-\left\lfloor (1-\delta)n \right\rfloor }$ to the output system and call this map $\Lambda'$.
This means that, letting $\Xi\coloneqq \Lambda(\rho^{\otimes n})$, we have another state $\Xi'$ as an output from $\Lambda'$ as $\Xi'=\Lambda'(\rho^{\otimes n}) = \Xi\otimes \tau_{\Gibbs}^{\otimes n-\left\lfloor (1-\delta)n \right\rfloor }$.

Now we construct the catalyst of the single-shot correlated-catalytic transformation from $\rho$ to $\rho'$ (Fig.~\ref{fig:catalyst supplemental}).
Let $R$ be a register (label) system spanned by $\{\ket{i}\}_{i=1}^n$ with a trivial Hamiltonian $H_R=I$, where $I$ is an identity operator.
Then we set catalyst $C$ as $S^{\otimes n-1}\otimes R$ and its state $c$ as
\eq{
c=\frac1n \sum_{k=1}^n \rho^{\otimes k-1}\otimes \Xi'_{n-k}\otimes \dm{k},
}
where $\Xi'_i$ is the reduced state of $\Xi'$ to the first $i$-th copies of $S$, that is, $\Xi'_i:=\Tr_{\overline{S^{\otimes i}}}[\Xi']$.
The initial state of the composite system $S\otimes C$ reads
\eq{
\rho\otimes c=\frac1n \sum_{k=1}^n \rho^{\otimes k}\otimes \Xi'_{n-k}\otimes \dm{k}.
}
Here, the last $n-1$ copies of $S$ are assigned to $C$.

Our single-shot correlated-catalytic transformation is constructed as follows.
We first apply $\Lambda'$ on $S\otimes C=S^{\otimes n}\otimes R$ conditioned on the label $\ket{n}$ on $R$, where the state on $S^{\otimes n}$ is $\rho^{\otimes n}$.
This---a thermal operation conditioned on the outcome of an incoherent measurement---can be implemented by a thermal operation~\cite[Proposition S.23]{Watanabe2024black}.
We also relabel the state of $R$ as $\ket{i}\to \ket{i+1}$ for $1\leq i\leq n-1$ and $\ket{n}\to \ket{1}$, which can be done by thermal operation noting that $H_R=I$.
Then, the resulting state is
\eq{
\pi=\frac1n \sum_{k=1}^n \rho^{\otimes k-1}\otimes \Xi'_{n-k+1}\otimes \dm{k}.
}
Here, by assigning the first $n-1$ copies of $S$ to $C$, the state of catalyst $c$ is recovered.
In addition, the state of the main system $S$ reads
\eq{
\Tr_{C}[\pi]=\frac1n \sum_{k=1}^n \Tr_{\overline{S_k}}[\Xi'_k]=(1-\delta)\frac{1}{(1-\delta)n} \( \sum_{k=1}^{(1-\delta)n} \Tr_{\overline{S_k}}[\Xi_k] \) +\delta \tau_{\Gibbs}.
}
Since we have
\eq{
\abs{\frac{1}{(1-\delta)n} \( \sum_{k=1}^{(1-\delta)n} \Tr_{\overline{S_k}}[\Xi_k] \) -\rho'}_1\leq \frac{1}{(1-\delta)n} \( \sum_{k=1}^{(1-\delta)n} \abs{\Tr_{\overline{S_k}}[\Xi_k]  -\rho'}_1\) <\ep
}
and
\eq{
\abs{\tau_{\Gibbs}-\rho'}_1\leq 2,
}
by setting $\ep=\delta=\frac{\ep}{4}$, we have the desired relation
\eq{
\abs{\Tr_{C}[\pi]-\rho'}<\ep .
}
\end{proof}


\end{document}